\documentclass[12pt]{article}

\usepackage{amssymb,amsmath,amsfonts,amssymb}
\usepackage{graphics,graphicx,color,hhline}
\usepackage{bbm} 
\usepackage{authblk}
\usepackage{euscript}
 \usepackage{mathtools}

\def\@abssec#1{\vspace{.05in}\footnotesize \parindent .2in 
{\bf #1. }\ignorespaces} 
\graphicspath{{/EPSF/}{Figures/}}   
\newtheorem{theorem}{Theorem}[section]
\newtheorem{lemma}[theorem]{Lemma}
\newtheorem{proposition}[theorem]{Proposition}

\newtheorem{remark}[theorem]{Remark}

\newtheorem{assumptions}[theorem]{Assumption}
\def\ds{\displaystyle}
\def \Rm {\mathbb R}
\def \Nm {\mathbb N}

\def\un{{\mathbbmss{1}}}

\newcommand{\be}{\begin{equation}}
\newcommand{\ee}{\end{equation}}
\newcommand{\bea}{\begin{eqnarray}}
\newcommand{\eea}{\end{eqnarray}}
\newcommand{\bee}{\begin{eqnarray*}}
\newcommand{\eee}{\end{eqnarray*}}

\def\fref#1{{\rm (\ref{#1})}}
\newcommand{\calQ}{\mathcal Q}

\newcommand{\calH}{\mathcal H}
\newcommand{\calL}{\mathcal L}

\newcommand{\calM}{\mathcal M}

\newcommand{\calE}{\mathcal E}

\newcommand{\calN}{\mathcal N}

\newcommand{\calS}{\mathcal S}
\newcommand{\calJ}{\mathcal J}
\newcommand{\calA}{\mathcal A}

\newcommand{\cout}[1]{}

\DeclareMathOperator{\Tr}{Tr}
\DeclarePairedDelimiter\bra{\langle}{\rvert}
\DeclarePairedDelimiter\ket{\lvert}{\rangle}
\begin{document}
\title{On local quantum Gibbs states}
\author{Romain  Duboscq \footnote{Romain.Duboscq@math.univ-tlse.fr}}
 \affil{Institut de Math\'ematiques de Toulouse ; UMR5219\\Universit\'e de Toulouse ; CNRS\\INSA, F-31077 Toulouse, France}
 \author{Olivier Pinaud \footnote{pinaud@math.colostate.edu}}
 \affil{Department of Mathematics, Colorado State University\\ Fort Collins CO, 80523}

\maketitle

\begin{abstract}
We address in this work the problem of minimizing quantum entropies under local constraints. We suppose macroscopic quantities such as the particle density, current, and kinetic energy are fixed at each point of $\Rm^d$, and look for a density operator over $L^2(\Rm^d)$ minimizing an entropy functional. Such minimizers are referred to as a local Gibbs states. This setting is in constrast with the classical problem of prescribing global constraints, where the total number of particles, total current, and total energy in the system are fixed. The question arises for instance in the derivation of fluid models from quantum dynamics. We prove, under fairly general conditions, that the entropy admits a unique constrained minimizer. Due to a lack of compactness, the main difficulty in the proof is to show that limits of minimizing sequences satisfy the local energy constraint. We tackle this issue by introducing a simpler auxiliary minimization problem and by using a monotonicity argument involving the entropy.
\end{abstract}
\maketitle

\section{Introduction} %\label{sec:intro}

This work is concerned with the study of quantum entropies of the form
$$
S(\varrho)=\Tr\big( \beta(\varrho)\big),
$$
where $\varrho$ is a density operator, namely a self-adjoint, trace class, positive operator on an infinite dimensional Hilbert space, $\Tr(\cdot)$ denotes operator trace, and $\beta$ is a convex function (note that we choose here the opposite of the convention traditionally used in the physics literature, for instance our definition of the von Neumann entropy is $\Tr\big( \varrho \log(\varrho)\big)$ and not its opposite). Quantum entropies have applications primarily in quantum information theory and quantum statistical physics, see e.g. \cite{ohya} for a review.

Our motivation here comes from quantum statistical physics and the construction of quantum statistical equilibria. These are important for instance in the description of reservoirs in the study of decohence in bipartite quantum systems, and in the derivation of fluids models, both classical and quantum. More precisely, we are interested in the following problem: given $S(\varrho)$ for an appropriate function $\beta$, we look for minimizers of $S(\varrho)$ under constraints of \textit{local} density, current, and kinetic energy. When the underlying Hilbert space is $L^2(\Rm^d)$, $d \geq 1$, these local quantities are informally defined as follows: if $\{\rho_p\}_{p \in \Nm}$ and $\{\phi_p\}_{p \in \Nm}$ denote the eigenvalues and eigenfunctions of a density operator $\varrho$, then
\be \label{consint}
\left\{
\begin{array}{ll}
\ds n[\varrho]=\sum_{p \in \Nm} \rho_p |\phi_p|^2, \qquad &\textrm{local density}\\
\ds u[\varrho] = \sum_{p\in\mathbb{N}} \rho_p \Im\left( \phi_p^* \nabla \phi_p\right),\qquad &\textrm{local current}\\ 
\ds k[\varrho]=\sum_{p \in \Nm} \rho_p |\nabla \phi_p|^2, \qquad &\textrm{local kinetic energy}.
\end{array}
\right.
\ee
For given functions $\{n_0,u_0,k_0\}$ in appropriate functional spaces, the problem consists in looking for minimizers of $S(\varrho)$ under the constraints that $\{n[\varrho],u[\varrho],k[\varrho]\}=\{n_0,u_0,k_0\}$.

Such quantum statistical equilibria arise in the work of Nachtergaele and Yau in \cite{nachter} in their derivation of the Euler equations from quantum dynamics. Therein, $S$ is the von Neumann entropy, and the minimizers are referred to as \textit{local Gibbs states}, which is the terminology we adopt here. Nachtergaele and Yau prove that a density operator $\varrho_t$ solution to the quantum Liouville equation
\be \label{liou}
i \partial_t \varrho_t=[H,\varrho_t], \qquad [H,\varrho_t]=H\varrho_t-\varrho_t H, \qquad H=-\Delta+V,
\ee
with as initial condition a local Gibbs state with constraints $\{n_0,u_0,k_0\}$, converges in an appropriate limit to a local Gibbs state with constraints $\{n_0(t),u_0(t),k_0(t)\}$, where $\{n_0(t),u_0(t),k_0(t)\}$ are the solutions to the Euler equations with initial condition $\{n_0,u_0,k_0\}$. In \cite{nachter}, the authors suggest that such local Gibbs states can be obtained by choosing appropriate Lagrange parameters, but there is no proof that this is indeed possible. The problem is actually not trivial, and our goal here is to establish that indeed the local Gibbs are well-defined and unique for a large class of entropy functions $\beta$ and adequate assumptions on the constraints. 

Local Gibbs states are also central in the work of Degond and Ringhofer on the derivation of quantum fluid models. Their theory consists in transposing to the quantum picture the moment closure method used in the derivation of classical fluid models. They consider the quantum Liouville equation \fref{liou} augmented with a collision operator $\calQ(\varrho)$ that drives the system to an equilibrium:
\be \label{liou2}
i \partial_t \varrho_t=[H,\varrho_t]+i \calQ(\varrho_t).
\ee
Note that in \cite{nachter}, the system is assumed to be initially in a statistical equilibrium, while in \fref{liou2} above the system converges in the long-time limit to such equilibrium. After deriving an infinite moment hierarchy from \fref{liou2}, Degond and coauthors close the system by introducing local Gibbs states and obtain various quantum fluids models such as the quantum Euler equations, quantum Navier-Stokes, or the quantum drift-diffusion model. See \cite{DR, QHD-review} for more details on this topic.

Our main contribution in this work is to establish that the entropy $S(\varrho)$ admits a unique minimizer under fairly general conditions on $\beta$ and on the constraints $\{n_0,u_0,k_0\}$. We work in the one-particle picture in the space $L^2(\Rm^d)$, for any $d \geq 1$, and plan in future work to extend our results to the many-body problem in the fermionic (or bosonic) Fock space. The latter is the actual problem that arises in \cite{nachter}, and is also of interest in the density functional theory at non-zero temperature, see e.g. \cite{thermalDFT}.

The main difficult in the proof consists in recovering the local energy constraint from limits of minimizing sequences of the entropy. It can be shown that such sequences converge to a density operator $\varrho_\star$ with finite energy, and as a consequence that $k[\varrho_\star]$ is well-defined. The crucial part is then to show that $\varrho_\star$ satisfies the energy constraint, i.e. $k[\varrho_\star]=k_0$. There is no sufficient compactness to directly pass to the limit in the local energy, and the sole straightforward information is that $\|k[\varrho_\star]\|_{L^1} \leq \|k_0\|_{L^1}$. We proved in \cite{DP-JMPA} that the equality $k[\varrho_\star]=k_0$ actually holds but the technique is limited to one-dimensional bounded domains. We introduce here a new argument allowing us to extend the result of \cite{DP-JMPA} to $\Rm^d$ for arbitrary $d \geq 1$. The key novel ingredient is to define an auxiliairy, simpler minimization problem with \textit{global} constraints instead of \textit{local} constraints. Namely, $\|n[\varrho]\|_{L^1}$, $\|u[\varrho]\|_{L^1}$ and  $\|k[\varrho]\|_{L^1}$ are prescribed instead of $n[\varrho]$, $u[\varrho]$, $k[\varrho]$. Such a minimization problem is shown to have a unique solution by introducing what we call \textit{generalized Gibbs states}, i.e. minimizers of the quantum free energy
$$
F_T(\varrho)=E(\varrho)+T S(\varrho)
$$
under a global density constraint, i.e. $\Tr(\varrho)$ is fixed. Above, $T>0$ is temperature, and $E(\varrho)=\Tr \big(\sqrt{H} \varrho \sqrt{H}\big)$ is the total energy with $H=-\Delta +V$ for an appropriate potential $V$. We prove the key result that such generalized Gibbs states have strictly monotone total energy and entropy with respect to the temperature. This eventually allows us to construct minimizers of the auxiliary problem and show the crucial property that their entropy is strictly monotone in the global kinetic energy constraint on $\|k[\varrho]\|_{L^1}$. These results are standard and intuitive in the context of classical thermodynamics, and are of independent interest for the generalized Gibbs states defined here. In some cases, we are also able to quantify precisely the dependency of the kinetic energy on the temperature. Using the strict monotonicity of the entropy, we are then able to go back to the original problem with local constraints and prove our result.

% Once the monotonicity of the entropy is established, the proof of compactness, or equivalently that $k[\varrho_\star]=k_0$ goes as follows: We start by showing that the equality $\|k[\varrho_\star]\|_{L^1} =\|k_0\|_{L^1}$ is actually enough to obtain that $k[\varrho_\star]=k_0$; the argument is based on fine compactness properties of trace class operators. The main step is then to prove that $\|k[\varrho_\star]\|_{L^1} < \|k_0\|_{L^1}$ is not possible. This follows by contradiction by using the monotonicity property and by comparing the entropy of the minimizing sequences and their limit.

The article is structured as follows: in Section \ref{sec:resu}, we introduce the functional setting, the problem,  and state our results. We prove two theorems: in the first one, we obtain existence and uniqueness of minimizers of the entropy under appropriate assumptions on the entropy and the constraints; in the second theorem, we show these assumptions are verified under mild conditions. The section is concluded by an overview of the proof. Section \ref{proofth1} is devoted to the proof of the first theorem, and Section \ref{proofth2} to that of the second theorem.

\paragraph{Acknowledgement.} OP is supported by NSF CAREER grant DMS-1452349.

\section{Main results} \label{sec:resu}

We start with some preliminaries. 
\subsection{Functional setting}
%\paragraph{Functional setting.}
 We denote by $L^r(\Rm^d)$, $d\geq 1$, $r\in [1,\infty]$, the usual Lebesgue spaces of complex-valued functions on $\Rm^d$, and by $H^{k}(\Rm^d)$ for $k \geq 1$ the standard Sobolev spaces. The symbol $\langle\cdot,\cdot\rangle$ is the Hermitian product on $L^2(\Rm^d)$ with the convention $\langle f,g\rangle=\int_{\Rm^d} f^* g dx$. The free Hamiltonian $-\Delta$, equipped with the domain $H^2(\Rm^d)$, is denoted by $H_0$. Moreover, $\calL(L^2(\Rm^d))$ is the space of bounded operators on $L^2(\Rm^d)$, and  $\calJ_1 \equiv \calJ_1(L^2(\Rm^d))$ is the space of trace class operators on $L^2(\Rm^d)$. In the sequel, we will refer to a density operator as a self-adjoint, trace class, positive operator on $L^2(\Rm^d)$. For $\varrho^*$ the adjoint of $\varrho$ and $|\varrho|=\sqrt{\varrho^* \varrho}$, we introduce the following space:
$$\calE=\left\{\varrho\in \calJ_1:\, \overline{\sqrt{H_0}|\varrho|\sqrt{H_0}}\in \calJ_1\right\},$$
where $\overline{\sqrt{H_0}|\varrho|\sqrt{H_0}}$ denotes the extension of the operator $\sqrt{H_0}|\varrho|\sqrt{H_0}$ to $L^2(\Rm^d)$. The domain of $\sqrt{H_0}$ is naturally $H^1(\Rm^d)$. We will  drop the extension sign in the sequel to ease notation.  The space $\calE$ is a Banach space when endowed with the norm
$$\|\varrho\|_{\calE}=\Tr \big(|\varrho| \big)+\Tr\big(\sqrt{H_0}|\varrho|\sqrt{H_0}\big),$$
where $\Tr(\cdot)$ denotes operator trace. 
The energy space is the following closed convex subspace of $\calE$:
$$\calE^+=\left\{\varrho\in \calE:\, \varrho\geq 0\right\}.$$
The eigenvalues $\{\rho_p\}_{p \in \Nm}$ of a density operator are counted with multiplicity, and form a nonincreasing sequence of positive numbers, converging to zero if the sequence is infinite. We will keep the notation $\{\rho_p\}_{p \in \Nm}$ for simplicity for the eigenvalues of a finite-rank density operator, with the convention that $\rho_p=0$ for $p>N$ for some $N$. 

%\paragraph{Constraints.} 
\subsection{The constraints} We  first note that the quantities introduced in \fref{consint} are well-defined according to the following remark. 
\begin{remark} \label{rem1} Let $\varrho \in \calE^+$ with eigenvalues $\{\rho_p\}_{p \in \Nm}$ and eigenvectors $\{\phi_p\}_{p \in \Nm}$. Then all series defined in \fref{consint} converge in $L^1(\Rm^d)$ and almost everywhere. Moreover,
$$
\Tr(\varrho)= \| n[\varrho]\|_{L^1}, \qquad \Tr\big(\sqrt{H_0}\,\varrho\sqrt{H_0}\big)=\| k[\varrho]\|_{L^1}.
$$
Note also that $\varrho \in \calE^+$ implies that $\nabla \sqrt{n[\varrho]} \in (L^2(\Rm^d))^d$ according to the inequality
$$
\|\nabla \sqrt{n[\varrho]}\|^2_{L^2} \leq \|\varrho\|_{\calE}.
$$ 
\end{remark}

\bigskip

The set of admissible local constraints is defined by
\begin{align*}
\calM = \left\{(n,u,k)\in L_+^1(\mathbb{R}^d)\times (L^1(\mathbb{R}^d))^d\times L_+^1(\mathbb{R}^d):\; \sqrt{n}\in H^1(\mathbb{R}^d)\right.
\\ \left. \qquad  \qquad\textrm{and}\quad (n,u,k) = (n[\varrho],u[\varrho],k[\varrho])\quad\textrm{for at least one }\varrho\in\calE^+ \right\}.
\end{align*}
Above, $L^1_+(\mathbb{R}^d) = \{ \varphi\in L^1(\mathbb{R}^d):\; \varphi\geq 0\; a.e.\}$. In other terms, $\calM$ consists of the set of functions $(n,u,k)$ that are the local density, current and kinetic energy of at least one density operator with finite energy. In the context of Degond-Ringhofer theory, the constraints are always in $\calM$ as originating from the solution $\varrho_t$ to the quantum Liouville equation \fref{liou2}. To the best of our knowledge, the characterization of $\calM$ remains to be done. We have though the following straightforward and helpful remark.

\begin{remark} \label{rem2} Constraints in the admissible set $\calM$ satisfy some compatibility conditions. 
Let indeed $\varrho\in\calE^+$. Then, 
\begin{equation} \label{ineqk}
k[\varrho]\geq  \frac{|u[\varrho]|^2}{n[\varrho]}+\left|\nabla \sqrt{n[\varrho]}\right|^2.
\end{equation}
To see this, we remark that 
\begin{align*} %\label{nabn}
\frac{1}{2}\nabla n[\varrho] = \sum_{p\in\mathbb{N}} \rho_p \Re\left( \phi_p^* \nabla \phi_p\right)\quad\textrm{and}\quad u[\varrho] = \sum_{p\in\mathbb{N}} \rho_p \Im\left( \phi_p^* \nabla \phi_p\right),
\end{align*}
where both series converge in $L^1(\Rm^d)$ and a.e. according to Remark \ref{rem1},
and, we deduce that
\begin{equation*}
\frac{1}{2}\nabla n[\varrho] + i u[\varrho] = \sum_{p \in \Nm} \rho_p \phi_p ^* \nabla \phi_p.
\end{equation*}
It follows, by the Cauchy-Schwarz inequality, that
\begin{align*}
\frac{1}{4}|\nabla n[\varrho]|^2+ |u[\varrho]|^2  &= \left| \frac{1}{2}\nabla n[\varrho] + i u[\varrho]\right|^2 = \left| \sum_{p \in \Nm} \rho_p \phi_p ^* \nabla \phi_p\right|^2
\\ &\leq \left(\sum_{p \in \Nm}\rho_p |\phi_p|^2\right)\left(\sum_{p \in \Nm}\rho_p |\nabla \phi_p|^2\right) = n[\varrho]k[\varrho],
\end{align*}
which yields \fref{ineqk}.
\end{remark}

% \begin{equation*}
%  k[\varrho] : = \sum_{j\in\mathbb{N}} \rho_j[\varrho] |\nabla \phi_j[\varrho]|^2 = - \sum_{j = 1}^d n[\partial_{x_j} \varrho \partial_{x_j}] = - n[\nabla\cdot\varrho\nabla].
% \end{equation*}
\bigskip

For $(n_0,u_0,k_0) \in \calM$, the feasible set is then given by
\begin{equation*}
\calA(n_0,u_0,k_0) = \left\{\varrho\in\calE^+:\; n[\varrho] = n_0,\; u[\varrho] = u_0\;\textrm{and}\; k[\varrho] = k_0 \right\}.
\end{equation*}
The set $\calA(n_0,u_0,k_0)$ is not empty by construction since  $(n_0,u_0,k_0)$ is admissible. We will use the notation $\bar{n}=\|n_0\|_{L^1}$ in the rest of the paper.

%\paragraph{Results.}
\subsection{Results}  We are now in position to state our main results. 
We recall that the entropy of a density operator in $\calE$ is denoted by
$$
S(\varrho)=\Tr \big( \beta(\varrho)\big),
$$
where $\beta$ satisfies the following assumption.
\begin{assumptions} \label{A0} (Convexity) The function $\beta \in C^0([0,\bar{n}])$ is strictly convex and verifies $\beta(0)=0$.
\end{assumptions} 

It is a standard result, see e.g. \cite{wehrl} section II, that the strict convexity of $\beta$ implies that of $S(\varrho)$. We present our results in two steps. First of all, we prove in Theorem \ref{mainth} that $S$ admits a unique minimizer in $\calA(n_0,u_0,k_0)$ under a set of assumptions on $\beta$ and $\calA(n_0,u_0,k_0)$. Assumptions \ref{A1}-\ref{A2} below are natural in that we expect the entropy to be bounded below and lower semi-continuous. Assumption \ref{A4} is the crucial ingredient allowing us to overcome the lack of compactness of the minimizing sequences. In Theorem \ref{mainth2}, we provide a set of conditions under which Assumptions \ref{A1}-\ref{A2}-\ref{A4} hold.

We state first the next two assumptions:

\begin{assumptions} \label{A1} (Boundedness from below) There exists a subset $\calM_0$ of $\calM$ such that, for $(n_0,u_0,k_0) \in \calM_0$, $S(\varrho)$ is well-defined on $\calA(n_0,u_0,k_0)$ (possibly infinite for some $\varrho$), and such that
  $$
    -\infty< \inf_{\calA(n_0,u_0,k_0)} S.
  $$
\end{assumptions}

\begin{assumptions} \label{A2} (Lower semi-continuity) Let $\{\varrho_k\}_{k \in \Nm}$ be a sequence in $\calA(n_0,u_0,k_0)$ such that $\varrho_k \to \varrho$ strongly in $\calJ_1$ as $k \to +\infty$. Then, 
  $$
  S(\varrho) \leq \liminf_{k \to +\infty} S(\varrho_k).
  $$
\end{assumptions}

In the next assumption, we suppose that the prescribed density $n_0$ decays sufficiently fast at the infinity. This will be needed in order to state Assumption \ref{A4}.

\begin{assumptions} \label{A3}(Confinement) There exists a nonnegative potential $V \in L^2_{\rm{loc}}(\Rm^d)$, with $V \to +\infty$ as $|x| \to +\infty$,  such that $n_0 V \in L^1(\Rm^d)$.
\end{assumptions}

For a potential $V$ as in Assumption \ref{A3}, we consider the self-adjoint operator $H=H_0+V$ defined on an appropriate domain $D(H)$, and introduce the total energy  
$$E(\varrho)=\Tr \big(\sqrt{H} \varrho \sqrt{H}\big).$$ 
Note that $E(\varrho)$ is finite when $\varrho \in \calA(n_0,u_0,k_0)$ and $n_0 V \in L^1(\Rm^d)$. The spectrum of $H$ is purely discrete and its ground state $\{\lambda_0,\phi_0\}$ is non-degenerate, see e.g. Theorems XIII.47 and XIII.67 in  \cite{RS-80-4}.
It can moreover easily be checked that for any $\varrho \in \calE^+$  with $n[\varrho] V \in L^1(\Rm^d)$, we have
$$
E(\varrho)=\Tr \big(\sqrt{H_0} \varrho \sqrt{H_0}\big) + \|n[\varrho] V\|_{L^1}.
$$
The introduction of $V$ is not necessary when the problem is posed on bounded domains of $\Rm^d$ since the free Hamiltonian $H_0$ has then a purely discrete spectrum (under appropriate boundary conditions).  We will need the following set of density operators with finite energy and fixed total trace: for $a>0$,
\begin{equation*}
\mathcal{S}(a) : = \left\{\varrho \in \calE^+; \; \Tr\left(\varrho\right)=a\right\}.
\end{equation*}
We state then the
\begin{assumptions} \label{A4}(Monotonicity w.r.t. temperature) There exists $T_c\geq 0$, such that for each $T>T_c$, the free energy $F_T(\varrho)=E(\varrho)+ T S(\varrho)$ admits a unique minimizer $\varrho_T$ on $\calS(\bar n)$ such that $u[\varrho_T]=0$,  $T \mapsto E(\varrho_T)$ (resp. $T \mapsto S(\varrho_T)$) is continuous strictly increasing (resp. decreasing), with $\lim_{T\to +\infty }E(\varrho_T)=+\infty$, and $\lim_{T\to T_c }E(\varrho_T)= \lambda_0 \, \bar n $, for $\lambda_0$ the smallest eigenvalue of $H$.
\end{assumptions}

A few comments are in order. Assumption \ref{A4} is used in the construction of minimizers of the entropy under \textit{global} constraints of density, current and energy, i.e., $\|n[\varrho]\|_{L^1}$, $\|u[\varrho]\|_{L^1}$ and  $\|k[\varrho]\|_{L^1}$ are prescribed instead of $n[\varrho]$, $u[\varrho]$, $k[\varrho]$. Such minimizers are the novel and key ingredient in obtaining sufficient  compactness to recover the local energy constraint. The monotonicity of the energy trivially holds in the context of classical thermodynamics for the Boltzmann entropy for instance, where the energy is linear in $T$. In the quantum case, it is not difficult to show that the energy of the usual Gibbs states
$$
\varrho_T= e^{-\frac{H}{T}} / Z_T, \qquad Z_T= \bar{n}^{-1} \Tr (e^{-\frac{H}{T}}),
$$
obtained for the Boltzmann entropy, is strictly increasing. For more general entropies, we will refer to the minimizers of $F_T$ as \textit{generalized Gibbs states}. They take the form $\varrho_T=\xi((H+\mu_T)/T)$, where $\xi(x)=(\beta')^{-1}(-x)$ and $\mu_T$ is the chemical potential ensuring the constraint $\Tr(\varrho_T)=\bar n$ is satisfied. It is only implicitly defined, while we have explicitly $\mu_T= T \log Z_T$ for Gibbs states. The implicit nature of $\mu_T $ makes the analysis more difficult, in particular that of the limit $T \to +\infty$ which requires some care. When the generalized Gibbs state has finite rank, which occurs if $\beta'(0)$ is finite, we will in particular need to resort to the semi-classical form of the Lieb-Thirring inequality (which is then an equality) to obtain that the energy tends to the infinity as $T \to +\infty$. The limit as $T \to T_c$ yields a minimizer of the energy alone. The introduction of a critical temperature $T_c$ is necessary since the energy is actually constant for $T \in (0,T_c]$ under some conditions on $\beta'$ when the generalized Gibbs state has finite rank.

\medskip

Our first result is then the following.

\begin{theorem} \label{mainth} Under Assumptions \ref{A0}-\ref{A1}-\ref{A2}-\ref{A3}-\ref{A4}, the entropy $S$ admits a unique minimizer in $\calA(n_0,u_0,k_0)$.
\end{theorem}

Only the existence and uniqueness of a minimizer $\varrho_\star$ is addressed in Theorem \ref{mainth}. Its characterization will be considered in future work, and appears to be quite difficult. The formal solution reads $\varrho_\star=\xi(\calH(A,B,C))$, where $\calH(A,B,C)$ is an Hamiltonian depending on the Lagrange parameters $(A,B,C)$ associated with the local constraints. The essential difficulty is to obtain sufficient regularity on the implicitly defined $(A,B,C)$ to give a sense to the expression $\varrho_\star=\xi(\calH(A,B,C))$. This was achieved in a one-dimensional setting for the von Neumann entropy in \cite{DP-CVPDE}, and in $\Rm^d$ for the density constraint only in \cite{DP-JFA}.

We show in our next result, and under minimal assumptions on $\beta$ and $V$, that Assumptions \ref{A1}-\ref{A2}-\ref{A4} are verified. We require slightly more regularity on $\beta$ than the one in Assumption \ref{A2}, in that strict convexity only implies that $\beta''(x)$ is nonnegative and defined a.e., while we suppose in addition that $\beta'' $ is continuous and strictly positive. The bound \fref{polcont} below allows $\beta'$ to blow up around the origin but at a sufficiently slow polynomial rate. This condition is sufficient to show that density operators with finite total energy have a finite entropy. Depending on $\beta'(0)$, there are two possibilities: when $\beta'(0)=-\infty$, then the generalized Gibbs states have infinite rank; when $\beta'(0)$ is finite, they have finite rank. In the latter case, we need the H\"older regularity type conditions \fref{assbet} around the origin to prove that the energy of the generalized Gibbs state goes to the infinity as $T \to +\infty$.

\begin{theorem} \label{mainth2} Suppose Assumptions \ref{A0} and \ref{A3} hold, with in addition $\beta \in C^2((0,\bar{n}))$ and $\beta''>0$ on $(0,\bar{n})$, and that there exist $\bar x \in (0,\bar{n})$ and $\gamma\in (\frac{d}{d+2},1)$ such that
  \be\label{polcont}
  \sup_{x \in [0, \bar x]} |x|^{-\gamma+1} |\beta'(x)|=C_{\bar x,\gamma}<\infty.
  \ee
Furthermore, we assume
\begin{itemize}
 \item when $\beta'(0)$ is infinite, that the $V$ of Assumption \ref{A3} verifies $V^{\frac{d}{2}-\frac{\gamma}{1-\gamma}} \in L^{1}(\Rm^d)$ (with $\frac{\gamma}{1-\gamma}-\frac{d}{2}>0$ since $\gamma\in (\frac{d}{d+2},1)$),
 \item when $\beta'(0)$ is finite, that we can take \begin{equation*}
 V(x)=1+|x|^\theta,
\end{equation*} 
with $(\frac{\gamma}{1-\gamma}-\frac{d}{2})\theta>d$ (so that $V^{\frac{d}{2}-\frac{\gamma}{1-\gamma}} \in L^{1}(\Rm^d))$, and that there exist $\underline x>0$ and $r>0$ such that
  \be \label{assbet}
  c_- x^r \leq \beta'(x)-\beta'(0) \leq c_+ x^r, \qquad \forall x \in [0, \underbar x],
  \ee
where $c_-$ and $c_+$ are positive constants.
\end{itemize}
Then, Assumptions \ref{A1}-\ref{A2}-\ref{A4} are satisfied.
\end{theorem}

We verify below that some entropies frequently encountered in the literature satisfy the assumptions of Theorem \ref{mainth2}.

\begin{itemize}
\item[-] \textit{The Boltzmann entropy} $\beta(x)=x \log (x) -x$. It is also referred to as the von Neumann entropy in the quantum case. It is strictly convex, with $\beta'(x)=\log x$ so that \fref{polcont} is satisfied for any $\gamma<1$. Then $\xi(x)=e^{-x}$, and the generalized Gibbs state has infinite rank. 

\item[-] \textit{The Fermi-Dirac entropy} $\beta(x)=x \log (x) +(1-x) \log (1-x)$. It is also referred to as the binary entropy in information theory. Besides classical thermodynamics, it arises when considering the von Neumann entropy of quasi-free states on CAR algebras, see e.g. Chapter 9 in \cite{Alicki}.   It is strictly convex, with $\beta'(x)=\log (\frac{x}{1-x})$ so that \fref{polcont} is also satisfied for any $\gamma<1$ and $\bar n \leq 1$. Then $\xi(x)=1/(e^x+1)$, and the generalized Gibbs state has infinite rank.

\item[-] \textit{The Bose-Einstein entropy} $\beta(x)=x \log (x) -(1+x) \log (1+x)$. As the Fermi-Dirac entropy, it arises in classical thermodynamics and e.g. when considering the von Neumann entropy of quasi-free states on CCR algebras.   It is strictly convex, with $\beta'(x)=\log (\frac{x}{1+x})$ and \fref{polcont} satisfied for any $\gamma<1$. Then $\xi(x)=1/(e^x-1)$, and the generalized Gibbs state has infinite rank. Note that $\lim_{x \to a} \beta'(x)<+\infty$ when $a \in (1,+\infty]$ (while the limit was infinite in the previous cases when $a=+\infty$ and $a=1$), which has some incidence on the proofs. 

\item[-] \textit{The Tsallis entropy} $\beta(x)= (q-1)^{-1} x^q$ for $q \in (0,1) \cup (1,\infty)$ (note that the Tsallis entropy is actually $S(\varrho)= (q-1)^{-1}(\Tr(\varrho^q)-1)$ but the extra constant term plays no role). It is used in information theory. It is  strictly convex, with $\beta'(x)=q (q-1)^{-1} x^{q-1}$ and \fref{polcont} satisfied for any $\gamma<1$. When $q>1$, then $\beta'(0)=0$, and the generalized Gibbs state has finite rank.  Estimate \fref{assbet} is trivially verified in that case. When $q<1$, the rank is infinite.
\end{itemize} 

The regularized Boltzmann entropy of the form
$
\beta_\eta(x)=\beta(x+\eta)-\beta(\eta)$, $\eta>0$,
is also sometimes used in practice to justify some calculations since its derivative $\beta'_\eta(x)=\log(x+\eta)$ is now bounded, see e.g. \cite{DP-CVPDE}. The corresponding minimizer has then finite rank (which grows to the infinity as $\eta \to 0$).

\subsection{Strategy of the proof} \label{strategy} We underline here the main steps of the proof. It starts with a minimizing sequence $\{\varrho_m\}_{m \in \Nm}$ in $\calA(n_0,u_0,k_0)$. Compactness results then allow us to obtain a $\varrho_\star$ with finite energy such that $\varrho_m \to \varrho_\star$ in $\calJ_1$ (along a subsequence still denoted $\{\varrho_m\}_{m \in \Nm}$) and such that $n[\varrho_m] \to n_0$, $u[\varrho_m] \to u_0$ in $L^1(\Rm^d)$ as $m \to +\infty$. Regarding the energy constraint, we can only deduce that $\|k[\varrho_\star] \|_{L^1} \leq \|k_0\|_{L^1}$, which is not enough to conclude that $k[\varrho_\star]=k_0$ and therefore that $\varrho_\star \in \calA(n_0,u_0,k_0)$.

To overcome this difficulty, we realize first that it is actually enough to show that $\|k[\varrho_\star] \|_{L^1} =\|k_0\|_{L^1}$ to obtain $k[\varrho_\star]=k_0$. This is a consequence of the positivity of $k[\varrho_\star]$ and of an argument of the type weak convergence plus convergence of the norm implies strong convergence in $\calJ_1$. The core to the proof is then to show that $\|k[\varrho_\star] \|_{L^1} < \|k_0\|_{L^1}$ is not possible.

The first step for this is the inequality
\begin{equation} \label{low}
S(\varrho_{\star}) \leq \inf_{\sigma\in\calA(n_0,u_0,k_0)} S(\sigma),
\end{equation}
which follows from the lower semi-continuity of the entropy. Intuitively, the above inequality should only be possible if $\|k[\varrho_\star] \|_{L^1} =\|k_0\|_{L^1}$. Indeed, if $\|k[\varrho_\star] \|_{L^1} < \|k_0\|_{L^1}$, then $\varrho_\star$ is a state with an equal or lower entropy than that of a minimizer of $S$ that has a strictly larger total kinetic energy (assuming the infimum is a minimum). With our definition of the entropy, i.e. the opposite of the physical entropy, and the common heuristics that the physical entropy is a measure of disorder, this yields a contradiction as we expect the state $\varrho_\star$ with strictly smaller energy to have a strictly larger (mathematical) entropy.  The main difficulty is then to make this argument rigorous.

The second step is to introduce two minimization problems with global constraints, which, as a consequence of \fref{low}, verify
\be \label{contS3}
\inf_{\sigma\in\calS(a_0,b_0,c_1)}S(\sigma) \leq  \inf_{\sigma\in\calS(a_0,b_0,c_0)} S(\sigma),
\ee
where
\begin{equation*}
\calS(a,b,c)=\left\{ \varrho \in \calE^+:\; \|n[\varrho]\|_{L^1}=a, \quad \|u[\varrho]\|_{L^1}=b, \quad \|e[\varrho]\|_{L^1}= c\right\},
\end{equation*}
with $e[\varrho]=k[\varrho]+n[\varrho] V$ ($V$ as in Assumption \ref{A3}), and 
$$
\|e_0\|_{L^1}=c_0, \qquad \|e[\varrho_\star]\|_{L^1}=c_1,\qquad  \|n_0\|_{L^1}=a_0, \qquad  \|u_0\|_{L^1}=b_0 \in \Rm^d.
$$
The goal is to prove a contradiction now from \fref{contS3} since we suppose that $c_1<c_0$. The main benefit in working with global constraints is that they are much easier to handle than the local ones. Though, it is unclear at this point that the infima in \fref{contS3} are actual minima since there is still a compactness issue to recover the global energy constraints. We solve this difficulty by introducing the generalized Gibbs states (GGS) defined in the introduction. We ignore the global current constraints at this point since they are taken care of by a simple change of gauge. We expect intuitively that the total energy of a GGS is monotone with the temperature, and therefore that there is a unique temperature associated with a given total energy. A minimizer of the entropy with global constraints is then a GGS with a well-chosen temperature.

We make this argument precise by proving the strict monotonicity of the total energy w.r.t the temperature, and show that the range of the total energy is $[\lambda_0 \bar n, +\infty)$. The lower bound is simply the minimal energy at zero temperature. The proof essentially relies on calculus of variations type arguments, with as main difficulty accounting for the chemical potential $\mu_T$. The proof that the energy tends to the infinity as $T \to +\infty$ in the finite rank case where $\beta'(0)$ is finite requires somewhat surprisingly a more involved strategy based on Riesz means for the operator $H=H_0+V$.

Once global minimizers are obtained, we show, using again the GGS, that the entropy is strictly monotone in the temperature and therefore in the total energy, which then proves the contradiction in \fref{contS3}.

\section{Proof of Theorem \ref{mainth}} \label{proofth1}
We start with some technical results.
\subsection{Preliminary technical lemmas}
The next Lemma is proved in \cite{MP-JSP}.
\begin{lemma}\label{lem:convcalE} (Compactness of bounded sequences in $\calE^+$)
Let $\{\varrho_k\}_{k\in\mathbb{N}}$ be a bounded sequence in $\mathcal{E}^+$. Then, there exists $\varrho\in\mathcal{E}^+$ and a subsequence such that
\begin{enumerate}
\item $\varrho_{k_m}\underset{m\to+\infty}\to\varrho$ in $\calJ_1$.
\item $\sqrt{H_0}\varrho_{k_m}\sqrt{H_0} \to \sqrt{H_0}\varrho\sqrt{H_0}$ weak-$*$ in $\calJ_1$, that is, for any compact operator $K$, 
\begin{equation*}
\Tr\left(K\sqrt{H}\varrho_{k_m}\sqrt{H_0} \right)\underset{m\to+\infty}\to \Tr\left(K\sqrt{H_0}\varrho\sqrt{H_0} \right).
\end{equation*}
\item $\Tr\left(\sqrt{H_0}\varrho\sqrt{H_0} \right) \leq \liminf_{m\to+\infty} \Tr\left(\sqrt{H_0}\varrho_{k_m}\sqrt{H_0} \right)$.
\end{enumerate}
%Furthermore, if $\Tr\left(\sqrt{H_0}\varrho\sqrt{H_0} \right) = \lim_{k\to+\infty} \Tr\left(\sqrt{H_0}\varrho_k\sqrt{H_0} \right)$, then we can conclude that
%\begin{equation*}
%\sqrt{H_0}\sqrt{\varrho_k}\underset{k\to+\infty}\to\sqrt{H_0}\sqrt{\varrho}\qua%d\textrm{in}\quad\calJ^2.
%\end{equation*}
\end{lemma}
The Lemma below is proved in \cite{Simon-trace}, Theorem 2.21 and addendum H.
\begin{theorem}\label{thm:AddH}
Suppose that $\varrho_m\to \varrho$ weakly in the sense of operators and that $\|\varrho_m\|_{\mathcal{J}_1}\underset{m\to\infty}{\to}\|\varrho\|_{\mathcal{J}_1}$. Then $\|\varrho_m-\varrho\|_{\mathcal{J}_1}\underset{m\to\infty}{\to}0$.
\end{theorem}

In the first step of the proof, we show that the problem is reduced to handling a global energy constraint.
\subsection{Step 1: reduction to the energy constraint}
According to Assumption \ref{A1}, the entropy $S$ is bounded from below on $\calA(n_0,u_0,k_0)$. There exists therefore a minimizing sequence $\{\varrho_m\}_{m\in\mathbb{N}}$ in $\calA(n_0,u_0,k_0)$ such that
\begin{equation*}
\lim_{m\to+\infty} S(\varrho_m) =\inf_{\sigma\in\calA(n_0,u_0,k_0)} S(\sigma) > -\infty.
\end{equation*}
Since by construction $\|n[\varrho_m]\|_{L^1}=\|n_0\|_{L^1}$ and $\|k[\varrho_m]\|_{L^1}=\|k_0\|_{L^1}$, it follows that $\{\varrho_m\}_{m\in\mathbb{\Nm}}$ is bounded in $\mathcal{E}^+$ and, hence, up to a subsequence (still abusively denoted by $\{\varrho_m\}_{m\in\mathbb{\Nm}}$), the sequence converges strongly in $\calJ_1$ to a $\varrho_{\star} \in \calE^+$ according to Lemma \ref{lem:convcalE}. Then, by Assumption \ref{A2}, we deduce the inequality \fref{low}.
%\begin{equation} \label{low}
%S(\varrho_{\star}) \leq \inf_{\sigma\in\calA(n_0,u_0,k_0)} S(\sigma).
%\end{equation}

It remains to prove that $\varrho_{\star}\in\calA(n_0,u_0,k_0)$. The fact that $n[\varrho_\star]=n_0$ is a direct consequence of the convergence of $\{\varrho_m\}_{m\in\mathbb{\Nm}}$ in $\calJ_1$,  and that $u[\varrho_\star]=u_0$ is established in \cite{MP-KRM}, Theorem 4.3. The latter follows from sufficient compactness as $u[\varrho_m]$ involves one less derivative than the (bounded) kinetic energy $\Tr (\sqrt{H_0}\varrho_m \sqrt{H_0})$. The remaining point is then to prove that $k[\varrho_\star]=k_0$, which is the essential difficulty in the proof. For this, we remark first from Item 3 of Lemma \ref{lem:convcalE} that
\begin{equation} \label{ineqinf}
\|k[\varrho_{\star}]\|_{L^1} = \Tr(\sqrt{H_0}\varrho_{\star}\sqrt{H_0}) \leq \liminf_{m\to+ \infty} \Tr(\sqrt{H_0}\varrho_m\sqrt{H_0}) = \|k_0\|_{L^1},
\end{equation}
and Item 2 that 
\begin{equation*}
\sqrt{H_0}\varrho_m\sqrt{H_0}\underset{m\to+\infty}{\to}\sqrt{H_0}\varrho\sqrt{H_0}\quad\textrm{weak-}*\;\textrm{in }\calJ_1.
\end{equation*}
The important observation is that it is enough to study $\|k[\varrho_{\star}]\|_{L^1} $ and $\|k_0\|_{L^1}$, and not the functions $k[\varrho_\star]$ and $k_0$ themselves. Indeed, if $\|k[\varrho_{\star}]\|_{L^1} = \|k_0\|_{L^1}$, then we claim that $k[\varrho_{\star}] = k_0$, and as a consequence $\varrho_{\star}$ is in $\calA(n_0,u_0,k_0)$ and the  minimizer is unique by strict convexity of the entropy. To prove the claim, if we have $\|k[\varrho_{\star}]\|_{L^1} = \|k_0\|_{L^1}$, we obtain from \fref{ineqinf} that
\bee
\|\sqrt{H_0}\varrho_{\star}\sqrt{H_0}\|_{\calJ_1}=\Tr(\sqrt{H_0}\varrho_{\star}\sqrt{H_0})&=&\lim_{m\to+ \infty} \Tr(\sqrt{H_0}\varrho_m\sqrt{H_0})\\ &=& \lim_{m\to+ \infty}\|\sqrt{H_0}\varrho_{m}\sqrt{H_0}\|_{\calJ_1},
\eee
and since the weak-$*$ convergence in $\calJ_1$ implies the weak convergence in the sense of operators, we deduce, by using Theorem \ref{thm:AddH}, that
\begin{equation*}
\sqrt{H_0}\varrho_m\sqrt{H_0}\underset{m\to+\infty}{\to}\sqrt{H_0}\varrho_{\star}\sqrt{H_0}\quad\textrm{in }\calJ_1.
\end{equation*}
Since $k[\varrho_m]=-n[\nabla \cdot \varrho_m \nabla]$, it can then directly be shown (by e.g. writing $k[\varrho_m]=n[\nabla \cdot  (\sqrt{H_0}+1)^{-1}(\sqrt{H_0}+1)\varrho_m (\sqrt{H_0}+1) (\sqrt{H_0}+1)^{-1} \nabla]$) that the above convergence implies that $k[\varrho_m] \to k[\varrho_\star]$ in $L^1(\Rm^d)$.

We therefore suppose that $\|k[\varrho_{\star}]\|_{L^1} < \|k_0\|_{L^1}$ and will prove a contradiction. For this, consider the $V$ of Assumptions \ref{A3}-\ref{A4}, and introduce the local total energy $e[\varrho]=k[\varrho]+V n[\varrho]$, which is in $L^1(\Rm^d)$ when $\varrho \in \calA(n_0,u_0,k_0)$ according to Assumption \ref{A3}. With $e_0 : = k_0 +V n_0 \in L^1(\Rm^d)$, let
\begin{equation*}
\calA_e(n_0,u_0,e_0) = \left\{\varrho\in\calE^+:\; n[\varrho] = n_0,\; u[\varrho] = u_0\;\textrm{and}\; e[\varrho] = e_0 \right\}.
\end{equation*}
Since both $k[\varrho]$ and $n[\varrho]$ are prescribed when $\varrho \in \calA(n_0,u_0,k_0)$, it follows that $\calA(n_0,u_0,k_0)=\calA_e(n_0,u_0,e_0)$, and will work from now on with $\calA_e(n_0,u_0,e_0)$ since it simplifies a bit the exposition. Note that we have $\|e[\varrho_{\star}]\|_{L^1} < \|e_0\|_{L^1}$ as a consequence of our hypothesis that $\|k[\varrho_{\star}]\|_{L^1} < \|k_0\|_{L^1}$.

\subsection{Step 2: reduction to global constraints}
With the notations $\calS(a,b,c)$, $a_0$, $b_0$, $c_0$ and $c_1$ of Section \ref{strategy}, % Let us first introduce the following notations:
% $$
% \|e_0\|_{L^1}=c_0, \qquad \|e[\varrho_\star]\|_{L^1}=c_1,\qquad  \|n_0\|_{L^1}=a_0, \qquad  \|u_0\|_{L^1}=b_0 \in \Rm^d.
% $$
% For any $(a,b,c)\in\mathbb{R}_+\times\mathbb{R}^d\times\mathbb{R}_+$, we define the following feasible set with global constraints
% \begin{equation*}
% \calS(a,b,c)=\left\{ \varrho \in \calE^+:\; \|n[\varrho]\|_{L^1}=a, \quad \|u[\varrho]\|_{L^1}=b, \quad \|e[\varrho]\|_{L^1}= c\right\},
% \end{equation*}
we define the set of admissible global constraints
\begin{equation*}
\calN(a,b,c) = \left\{(n,u,e)\in\calM_e:\; \|n\|_{L^1} = a,\;\|u\|_{L^1} = b,\;\textrm{and}\;\|e\|_{L^1} = c \right\},
\end{equation*}
where we write $(n,u,e)\in\calM_e$ for $(n,u,e-nu)\in\calM$ (we recall that $\calM$ is the set of admissible local constraints). Note that without more hypotheses on $(a,b,c)$, the sets $\calS(a,b,c)$ and $\calN(a,b,c)$ might be empty. Though, it is clear that $\varrho_\star$ belongs to $\calS(a_0,b_0,c_1)$ since $n[\varrho_\star]=n_0$ and $u[\varrho_\star]=u_0$. Hence, we have from \fref{low},
\be \label{contS2}
\inf_{\sigma\in\calS(a_0,b_0,c_1)}S(\sigma) \leq S(\varrho_{\star}) \leq \inf_{\sigma\in\calA_e(n_0,u_0,e_0)} S(\sigma).
\ee
We have then the following important result.
\begin{lemma} \label{infinf}
Let $a,b,c\in\mathbb{R}_+^{*}\times\left(\mathbb{R}^*\right)^d\times \mathbb{R}^{*}_+$ such that $\calS(a,b,c)$ is not empty, and assume that the minimization problem
\begin{equation*}
\inf_{\sigma\in\calS(a,b,c)} S(\sigma),
\end{equation*}
admits a unique solution.  Then, the following equality holds
\begin{equation} \label{lemresu}
\inf_{(n,u,e)\in\calN(a,b,c)}\left( \inf_{\sigma\in\calA_e(n,u,e)} S(\sigma)\right)= \inf_{\sigma\in\calS(a,b,c)} S(\sigma).
\end{equation}
\end{lemma}
\begin{proof}
We denote by $\varrho_0$ the solution of the minimization problem
\begin{equation*}
\inf_{\sigma\in\calS(a,b,c)} S(\sigma) = S(\varrho_0).
\end{equation*}
The first step consists in establishing equality \fref{eq:eqSlocalGlobal} below. For this, since 
\be \label{abc}(\|n[\varrho_{0}]\|_{L^1} ,\|u[\varrho_{0}]\|_{L^1},\|e[\varrho_{0}]\|_{L^1}) = (a,b,c),\ee
 we immediately have that $\mathcal{A}_e(n[\varrho_{0}] ,u[\varrho_{0}],e[\varrho_{0}])\subset\calS(a,b,c)$. Hence, we deduce that
\begin{equation} \label{ineq0}
\inf_{\sigma\in\calS(a,b,c)} S(\sigma)  \leq  \inf_{\sigma\in\calA_e(n[\varrho_{0}] ,u[\varrho_{0}],e[\varrho_{0}])} S(\sigma).
\end{equation}
Furthermore, we clearly have $\varrho_0\in \mathcal{A}_e(n[\varrho_{0}] ,u[\varrho_{0}],e[\varrho_{0}])$ and, thus, we also obtain
\begin{equation*}
 \inf_{\sigma\in\calA_e(n[\varrho_{0}] ,u[\varrho_{0}],e[\varrho_{0}])} S(\sigma) \leq S(\varrho_0).
\end{equation*}
Together with \fref{ineq0}, this yields
\begin{equation}\label{eq:eqSlocalGlobal}
\inf_{\sigma\in\calS(a,b,c)} S(\sigma) =  \inf_{\sigma\in\calA_e(n[\varrho_{0}] ,u[\varrho_{0}],e[\varrho_{0}])} S(\sigma).
\end{equation}

We now proceed to the proof of \fref{lemresu} by proving two opposite inequalities. Consider the mapping $L:\mathcal{M}_e\to\mathbb{R}$ defined by
\begin{equation*}
L(n,u,e) : = \inf_{\sigma\in\calA_e(n,u,e)} S(\sigma).
\end{equation*}
For the first inequality in \fref{lemresu}, we obtain directly, using \fref{abc} together with \eqref{eq:eqSlocalGlobal}, 
\begin{equation} \label{fineq}
\inf_{(n,u,e)\in\calN(a,b,c)} L(n,u,e) \leq L(n[\varrho_0],u[\varrho_0],e[\varrho_0]) = \inf_{\sigma\in\calS(a,b,c)} S(\sigma).
\end{equation}
For the second inequality, it follows, since for any $(n,u,e)\in\calN(a,b,c)$, we have $\calA_e(n,u,e)\subset \calS(a,b,c)$, that
$$
\inf_{\sigma\in\calS(a,b,c)} S(\sigma)  \leq   \inf_{\sigma\in\calA_e(n,u,e)} S(\sigma).
$$
Taking the infimum over $(n,u,e)\in\calN(a,b,c)$ yields 
\begin{equation*}
 \inf_{\sigma\in\calS(a,b,c)} S(\sigma)\leq\inf_{(n,u,e)\in\calN(a,b,c)}\left( \inf_{\sigma\in\calA_e(n,u,e)} S(\sigma)\right),
\end{equation*}
which gives the desired result together with \fref{fineq}.
\end{proof}

\bigskip

Note that both $\calS(a_0,b_0,c_1)$ and $\calS(a_0,b_0,c_0)$ are not empty by construction. Assuming $S$ admits unique minimizers on $\calS(a_0,b_0,c_1)$ and $\calS(a_0,b_0,c_1)$ (these facts will be proved in the next section), it follows from the previous Lemma and \eqref{contS2}, by taking the infimum over functions $(n_0,u_0,e_0)\in \calN(a_0,b_0,c_0)$, that
\begin{equation}\label{eq:IneqContraS}
\inf_{\sigma\in\calS(a_0,b_0,c_1)} S(\sigma)\leq \inf_{\sigma\in\calS(a_0,b_0,c_0)} S(\sigma).
\end{equation}
This inequality is the reason for the introduction of the minimization problem with global constraints, and we will prove that it cannot hold for $c_1<c_0$. This is based on Assumption \ref{A4} and the monotonicity of entropy in the temperature. We show first in the next step that $S$ admits unique minimizers on $\calS(a_0,b_0,c_0)$ and $\calS(a_0,b_0,c_1)$, as required in Lemma \ref{infinf}.
\subsection{Step 3: the global minimization problem}
We prove here the following result.
\begin{proposition} \label{exglob}Let $c \in [a_0 \lambda_0+|b_0|^2/a_0, +\infty)$ ($\lambda_0$ defined in Assumption \ref{A4}). Then $S$ admits a unique minimizer on $\calS(a_0,b_0,c)$.
\end{proposition}

The proof relies on the monotonicity of the energy stated in Assumption \ref{A4}.

\paragraph{Proof of Proposition \ref{exglob},  step 1: an auxiliary problem.} The proof begins by showing that the solution to the minimization problem
\begin{equation}\label{eq:minGdensitycurrent}
\min_{\varrho\in\calS(a_0,b_0)} F_T(\varrho)=\min_{\varrho\in\calS(a_0,b_0)} E(\varrho)+T S(\varrho ),
\end{equation}
where
\begin{equation*}
\calS(a_0)=\{ \varrho \in \calE^+:\; \|n[\varrho]\|_{L^1}=a_0\}, \qquad \calS(a_0,b_0)=\{ \varrho \in \calS(a_0): \; \| u[\varrho]\|_{L^1}=b_0\},
\end{equation*}
can be expressed in terms of the minimizers of $F_T$ on $\calS(a_0)$ thanks to a gauge transformation. We will need the following simple result.
\begin{lemma}\label{lem:GaugeB}
Let $a_0>0$ and $b,b_0\in\mathbb{R}^d$. Consider the transformation $G_b:\calE^+ \mapsto \calE^+$ given by
\begin{equation*}
G_b(\varrho) = e^{ix\cdot b} \varrho  e^{-ix\cdot b}.
\end{equation*} Then, for $\varrho\in\calS(a_0,b_0)$, we have
\begin{equation*}
n[G_b(\varrho)] = n[\varrho],\quad u[G_b(\varrho)] = u[\varrho] + n[\varrho] b\quad\textrm{and}\quad e[G_b(\varrho)] = e[\varrho] - 2 b \cdot u[\varrho] + |b|^2n[\varrho].
\end{equation*}
 \end{lemma}
 \begin{proof}
Denoting by $\{\rho_j,\phi_j\}_{j\in\mathbb{N}}$ the eigen-elements of $\varrho$, $G_b(\varrho)$ has the following spectral decomposition (all series below converge since $\varrho \in \calE^+$)
  \begin{equation*}
e^{ix\cdot b} \varrho  e^{-ix\cdot b}= \sum_{j\in\mathbb{N}} \rho_j |e^{ix\cdot b}\phi_j\rangle\langle e^{ix\cdot b}\phi_j|.
 \end{equation*}
 We directly deduce that $n[G_b(\varrho)] = n[\varrho]$ and we have
 \begin{equation*}
 u[G_b(\varrho)] = \sum_{j\in\mathbb{N}} \rho_j\Im\left(\phi_j^* (  i b\,\phi_j + \nabla \phi_j)\right) = b\, n[\varrho]+ u[\varrho].
 \end{equation*}
Moreover, we compute
 \begin{align*}
 k[G_b(\varrho)] &= \sum_{j\in\mathbb{N}} \rho_j \left|\nabla \phi_j + i b \phi_j \right|^2
 \\ &= k[\varrho] + \sum_{j\in\mathbb{N}} \rho_j \left(-2\Re(i  \phi_j^* \,b\cdot \nabla \phi_j)  + |b|^2\left|\phi_j \right|^2 \right) = k[\varrho] - 2 b \cdot u[\varrho] + |b|^2n[\varrho],
 \end{align*}
 which yields
 \begin{equation*}
 e[G_b(\varrho)] = e[\varrho] - 2 b \cdot u[\varrho] + |b|^2n[\varrho].
 \end{equation*}
 \end{proof}

\bigskip

Using the previous Lemma, and the fact that the eigenvalues of a density operator are not changed by the action of $G_{b_0/a_0}$, we have, for any $\varrho\in\calS(a_0,0)$,
\begin{equation*}
S(G_{b_0/a_0}(\varrho)) = S(\varrho)\quad\textrm{and}\quad E(G_{b_0/a_0}(\varrho)) = E(\varrho) + |b_0|^2/a_0.
\end{equation*}
Moreover, it is clear that $G_{b_0/a_0}$ is a bijective mapping from $\calS(a_0,0)$ to $\calS(a_0,b_0)$ with inverse $G_{-b_0/a_0}$. For any $T>0$, denote now by $\varrho_{T,a_0}$  the minimizer of the free energy $F_T$ in $\calS(a_0)$, which exists and is unique according to Assumption \ref{A4}. It verifies $u[\varrho_{T,a_0}]=0$. Hence, $\varrho_{T,a_0}\in\calS(a_0,0)\subset\calS(a_0)$. Thus, we have, using Lemma \ref{lem:GaugeB},
\begin{align*}
\min_{\varrho\in\calS(a_0,b_0)}F_T(\varrho) &= \min_{\sigma\in\calS(a_0,0)}F_T(G_{b_0/a_0}(\sigma)) = \min_{\sigma\in\calS(a_0,0)}F_T(\sigma) + |b_0|^2/a_0 
\\ &\geq \min_{\sigma\in\calS(a_0)}F_T(\sigma) + |b_0|^2/a_0 =   F_T(\varrho_{T,a_0}) + |b_0|^2/a_0  = F_T(G_{b_0/a_0}(\varrho_{T,a_0})),
\end{align*}
which proves that $G_{b_0/a_0}(\varrho_{T,a_0})$ is a minimizer of $F_T$ since $G_{b_0/a_0}(\varrho_{T,a_0}) \in \calS(a_0,b_0)$. Since $S$ is strictly convex, we deduce that $G_{b_0/a_0}(\varrho_{T,a_0})$ is the unique minimizer. We have therefore just characterized the solution to \fref{eq:minGdensitycurrent}.

\paragraph{Proof of Proposition \ref{exglob},  step 2: monotonicity argument.} We will use now Assumption \ref{A4}, and remark first that, by Lemma \ref{lem:GaugeB},
\begin{align*}
\mathrm{E}_{a_0,b_0}(T) : &=\Tr\left(\sqrt{H}G_{b_0/a_0}(\varrho_{T,a_0})\sqrt{H} \right) = \Tr\left(\sqrt{H}\varrho_{T,a_0}\sqrt{H} \right) + |b_0|^2/a_0 
\\ &=: \mathrm{E}_{a_0}(T)+ |b_0|^2/a_0.
\end{align*}
According to Assumption \ref{A4},  $\mathrm{E}_{a_0}(T)$ is a continuous strictly increasing function from $[T_c,+\infty)$ to $[a_0 \lambda_0,+\infty)$. Hence, for any $c\in[a_0\lambda_0+|b_0|^2/a_0 ,+\infty)$, there exists a unique $T(a_0,b_0,c)\geq T_c$ such that $G_{b_0/a_0}(\varrho_{T(a_0,b_0,c),a_0})$ has total energy $\mathrm{E}_{a_0,b_0}(T)=c$, and as a consequence
\begin{equation*}
G_{b_0/a_0}(\varrho_{T(a_0,b_0,c),a_0})\in \calS(a_0,b_0,c).
\end{equation*}
We now show that $G_{b_0/a_0}(\varrho_{T(a_0,b_0,c),a_0})$ is the unique minimizer of $S$ on $\calS(a_0,b_0,c)$. Since $G_{b_0/a_0}$ is a bijective mapping from $\calS(a_0,0,c-|b_0|^2/a_0)$ to $\calS(a_0,b_0,c)$, it follows that $\varrho_{T(a_0,b_0,c),a_0}\in\calS(a_0,0,c-|b_0|^2/a_0)$. Introducing the notation $T_0 = T(a_0,b_0,c)$, we have
\begin{align*}
T_0 S(\varrho_{T_0,a_0}) + c-|b_0|^2/a_0&=F_{T_0}(\varrho_{T_0,a_0}) =\min_{\varrho\in\calS(a_0)} F_{T_0}(\varrho)\leq \min_{\varrho\in\calS(a_0,b_0,c)} F_{T_0}(\varrho).
\end{align*}
Then
\begin{align*}
 \min_{\varrho\in\calS(a_0,b_0,c)} F_{T_0}(\varrho)&=  T_0 \min_{\varrho\in\calS(a_0,0,c-|b_0|^2/a_0)}S(G_{b_0/a_0}(\varrho)) + c \\& \leq T_0 S(G_{b_0/a_0}(\varrho_{T_0,a_0})) + c =T_0 S(\varrho_{T_0,a_0}) + c.
\end{align*}
Above, we used the facts that $S(G_{b_0/a_0}(\varrho_{T_0,a_0}))=S(\varrho_{T_0,a_0})$ since $G_{b_0/a_0}$ is unitary, and that $\varrho_{T_0,a_0} \in \calS(a_0,0,c-|b_0|^2/a_0)$ by construction. Hence, we obtain that \be \label{carS}
\displaystyle\min_{\varrho\in\calS(a_0,b_0,c)} S(\varrho) =  S(\varrho_{T_0,a_0}) = S(G_{b_0/a_0}(\varrho_{T_0,a_0}))\ee
and, by the strict convexity of $S$, that $G_{b_0/a_0}(\varrho_{T_0,a_0})$ is the unique minimizer of $S$ in $\calS(a_0,b_0,c)$. This concludes the proof of Proposition \ref{exglob}.

\subsection{Step 4: Conclusion}
We now use Proposition \ref{exglob} and need to show that $c_1 \geq a_0 \lambda_0 + |b_0|^2/a_0$ (which implies $c_0>a_0 \lambda_0 + |b_0|^2/a_0$ since by assumption $c_0>c_1$), in order to obtain that $S$ admits unique minimizers on $\calS(a_0,b_0,c_1)$ and $\calS(a_0,b_0,c_0)$. For this, we recall that  $\varrho_\star \in \calS(a_0,b_0,c_1)$, and that any density operator in $\calE^+$ verifies a.e., see Remark \ref{rem2},
$$
|\nabla \sqrt{n[\varrho]}|^2 + \frac{|u[\varrho]|^2}{n[\varrho]} \leq k[\varrho].
$$
Since we know by construction that $n[\varrho_\star]=n_0$ and $u[\varrho_\star]=u_0$, it follows that
$$
\|\nabla \sqrt{n_0}\|^2_{L^2}+\| u_0/\sqrt{n_0}\|^2_{L^2}+\|n_0 V\|_{L^1} \leq \|k[\varrho_\star]\|_{L^1}+\|n_0 V\|_{L^1}=\Tr (\sqrt{H}\varrho_\star \sqrt{H})=c_1.
$$
It is clear by the min-max principle that
$$
\lambda_0 \|n_0\|_{L^1} \leq \|\nabla \sqrt{n_0}\|^2_{L^2}+\|n_0 V\|_{L^1}=E(\ket{\sqrt{n_0}}\bra{\sqrt{n_0}}),
$$
and by the Cauchy-schwarz inequality that
$$
\frac{|b_0|^2}{a_0}=\frac{|\int_{\Rm^d} u_0(x) dx|^2}{\|n_0\|_{L^1}} \leq \| u_0/\sqrt{n_0}\|^2_{L^2}.
$$
Collecting inequalities, we find the desired result $c_1 \geq a_0 \lambda_0 + |b_0|^2/a_0$. As a consequence of Proposition \ref{exglob}, minimizers of $S$ on $\calS(a_0,b_0,c_1)$ and $\calS(a_0,b_0,c_0)$ are uniquely defined, and inequality \fref{eq:IneqContraS} holds true. It reads
$$
\min_{\sigma\in\calS(a_0,b_0,c_1)} S(\sigma)\leq \min_{\sigma\in\calS(a_0,b_0,c_0)} S(\sigma),
$$
which is also, using \fref{carS},
\begin{equation} \label{cont}
  S(\varrho_{T_1,a_0}) \leq S(\varrho_{T_0,a_0}),
\end{equation}
where $T_1 : = T(a_0,b_0,c_1)$, $T_0:=T(a_0,b_0,c_0)$ are obtained as in Step 2 of Proposition \ref{exglob}. Since $T(a_0,b_0,c)$ is strictly increasing with respect to $c$ as $T \mapsto \mathrm{E}_{a_0,b_0}(T)$ is strictly increasing, we have that $T_1<T_0$ since $c_1<c_0$. According to Assumption \ref{A4}, $T \mapsto S(\varrho_{a_0,T})$ is strictly decreasing, resulting in $S(\varrho_{T_0,a_0}) <S(\varrho_{T_1,a_0})$, which contradicts \eqref{cont}. The only possibility is therefore that $\|k[\varrho_\star]\|_{L^1}=\|k_0\|_{L^1}$ which concludes the proof of Theorem \ref{mainth}.

\section{Proof of Theorem \ref{mainth2}} \label{proofth2}

We start with Assumptions \ref{A1} and \ref{A2}, whose proofs are relatively direct, and turn next to Assumption \ref{A4} which requires more work.

  \subsection{Verification of Assumptions \ref{A1} and \ref{A2}}

The fact that Assumptions \ref{A1} and \ref{A2} are satisfied is a direct consequence of the next lemma.

\begin{lemma}\label{lem:EntropyProp}
  Under the conditions and notations of Theorem \ref{mainth2}, the following results hold:
  \begin{enumerate}
\item (Control of the entropy) Let $\varrho \in \calE^+$ with $n[\varrho] V \in L^1(\Rm^d)$. Then, we have the estimate
\begin{equation}\label{eq:estEntropy} 
\Tr(|\beta(\varrho)|) \leq C \max_{x \in [0,\bar{n}]} |\beta(x)|+C E(\varrho)^\gamma,
\end{equation}
where we recall $E(\varrho) =  \Tr( \sqrt{H} \varrho \sqrt{H })$.
\item (Continuity of the entropy) Let $\{\varrho_m\}_{m\in\mathbb{N}}$ be a bounded sequence in $\mathcal{E}^+$ with $\|n[\varrho_m] V\|_{L^1}$ bounded, that converges strongly in $\calJ_1$ to $\varrho \in \calE^+$ with $n[\varrho] V \in L^1(\Rm^d)$. Then,
  $$
  \lim_{m \to \infty} \Tr \left(\beta(\varrho_m)\right)=\Tr \left(\beta(\varrho)\right).
  $$
\end{enumerate}
\end{lemma}

The proof of the Lemma is postponed to the end of the section. We now verify Assumptions \ref{A1} and \ref{A2}. Consider for $\calM_0$ the set of $(n_0,u_0,k_0) \in \calM$ such that  $n_0 V \in L^1(\Rm^d)$ for $V$ as in Assumption \ref{A3} and satisfying $V^{-1} \in L^{\frac{\gamma}{1-\gamma}-\frac{d}{2}}(\Rm^d)$. Then, the first Item of Lemma \ref{lem:EntropyProp} holds under the conditions of Theorem \ref{mainth2}, and $|\Tr(\beta(\varrho))|$ is finite for $\varrho \in \calA(n_0,u_0,k_0)$. This  shows that Assumption \ref{A1} is verified. Regarding Assumption \ref{A2}, we first notice that if  $\{\varrho_m\}_{m\in\mathbb{N}}$ is a sequence in $\calA(n_0,u_0,k_0)$ converging in $\calJ_1$ to some $\varrho$, we have necessarily $n[\varrho] V \in L^1(\Rm^d)$ as a consequence of $n_0 V \in L^1(\Rm^d)$. Indeed, first of all $\sqrt{n[\varrho_m] V}$ is uniformly bounded in $L^2(\Rm^d)$ since $n[\varrho_m]=n_0$ and $n_0 V \in L^1(\Rm^d)$. Hence, there exists $v \in L^2(\Rm^d)$ and a subsequence such that $\sqrt{n[\varrho_{m_k}] V}$ converges weakly to $v$ with
$$
\|v\|_{L^2} \leq \liminf_{k \to \infty} \|n[\varrho_{m_k}] V\|^{1/2}_{L^1}=\|n_0V\|^{1/2}_{L^1}.
$$
To identify $v$, we remark first that $n[\varrho_m]$ converges to $n[\varrho]$ in $L^1(\Rm^d)$ since $\varrho_m$ converges to $\varrho$ in $\calJ_1$, and as a consequence $\sqrt{n[\varrho_m]} \to \sqrt{n[\varrho]}$ in $L^2(\Rm^d)$. Hence, for any $\varphi \in L^2(\Rm^d)$,
$$
\lim_{k \to \infty} \int_{\Rm^d} \sqrt{n[\varrho_{m_k}] V}(1+V)^{-1} \varphi dx = \int_{\Rm^d} v(1+V)^{-1} \varphi dx=\int_{\Rm^d} \sqrt{n[\varrho]V} (1+V)^{-1} \varphi dx,
$$
which yields $v=\sqrt{n[\varrho] V}$. Hence, $n[\varrho ] V\in L^1(\Rm^d)$. 

To conclude, since every sequence in $\calA(n_0,u_0,k_0)$ is bounded in $\calE^+$, we can use Item 2 of Lemma \ref{lem:EntropyProp} to obtain the continuity of the entropy. This yields Assumption \ref{A2}.

\paragraph{Proof of Lemma \ref{lem:EntropyProp}: Item 1. }  We will use the following classical inequality, see e.g. \cite[Lemma A.1]{MP-JSP}: for any $\varrho \in \calE^+$ such that $n[\varrho] V \in L^1(\Rm^d)$, we have
\begin{equation} \label{contsum}
\sum_{j\in\mathbb{N}} \lambda_j\rho_j \leq  \Tr\left( \sqrt{H}\varrho \sqrt{H}\right),
\end{equation}
where $\{\lambda_j\}_{j \in \Nm}$ is the nondecreasing sequence of eigenvalues of $H$ and $\{\rho_j\}_{j \in \Nm}$ the nonincreasing sequence of eigenvalues of $\varrho$. Let $a \in (0, \bar x]$. Then
\begin{equation*}
  \Tr \left(|\beta (\varrho)|\right)=\sum_{j\in\mathbb{N}} |\beta(\rho_j)| = \sum_{j <N_a} |\beta(\rho_j)|+\sum_{j \geq N_a} |\beta(\rho_j)|,
\end{equation*}
where $N_a$ is such that $\rho_{N_a} \leq a$ and $\rho_{N_a-1} > a$. Since $\beta(0)=0$, we have according to \fref{polcont}, for $x \in [0, \bar x]$,
\be \label{contbet}
|\beta(x)| \leq \int_0^x |\beta'(s)| ds \leq C |x|^{\gamma}. 
\ee
Using H\"older's inequality and \fref{contbet}, it follows that 
\bee
\sum_{j \geq N_a} |\beta(\rho_j)| &\leq& C \sum_{j \geq N_a} |\rho_j|^\gamma  = C \sum_{j\geq N_a} |\rho_j\lambda_j|^{\gamma} \lambda_j^{-\gamma} \\
&\leq& \left( \sum_{j\geq N_a} \rho_j\lambda_j\right)^{\gamma}\left( \sum_{j\geq N_a} \lambda_j^{-\frac{\gamma}{1-\gamma}}\right)^{1-\gamma}.
\eee
According to \cite{Dolbeault-Loss}, Theorem 1, we have the estimate
\be \label{sumeig}
\sum_{j\in \Nm} \lambda_j^{-\frac{\gamma}{1-\gamma}} \leq C_{\gamma,d} \int_{\Rm^d} (V(x))^{\frac{d}{2}-\frac{\gamma}{1-\gamma}}dx,
\ee
which is finite by the assumption on $V$ since $\frac{d}{2}-\frac{\gamma}{1-\gamma}<0$. Since we have directly
$$
\sum_{j <N_a} |\beta(\rho_j)| \leq N_a \max_{x \in [0, \bar x]} |\beta(x)|,
$$
we find, using \fref{contsum} and collecting estimates,
$$
 \sum_{j\in\mathbb{N}} |\beta(\rho_j)| \leq N_a \max_{x \in [0,\bar{n}]} |\beta(x)|+C E(\varrho)^\gamma,
$$
which is \fref{eq:estEntropy}.
\paragraph{Proof of Lemma \ref{lem:EntropyProp}: Item 2.} We now turn to the convergence of the entropy. Let $\eta \in (0,\bar n)$ and decompose
\begin{equation*}
\beta(t) = \beta(t)\un_{t\leq\eta} + \beta(t)\un_{t> \eta} = : \beta^{(1)}_{\eta}(t) + \beta^{(2)}_{\eta}(t).
\end{equation*}
Since $\{\varrho_m\}_{m\in\mathbb{N}}$ converges to $\varrho$ in $\mathcal{J}_1$, we have, for any $j\in\mathbb{N}$,
\begin{equation} \label{conveig}
\rho_{m,j}\underset{m\to+\infty}{\to} \rho_j,
\end{equation}
for $\{\rho_{m,j}\}_{j \in \Nm}$ the eigenvalues of $\varrho_m$. By continuity of $\beta_\eta^{(2)}$, this yields
\be \label{combet2}
\lim_{m \to \infty }\Tr\left(\beta^{(2)}_{\eta}(\varrho_m)\right)=\Tr\left(\beta^{(2)}_{\eta}(\varrho)\right).
\ee
Regarding $\beta_\eta^{(1)}$, let $N_\eta$ and $N_\eta^m$  such that $\rho_{N_\eta}\leq \eta$, $\rho_{N_\eta-1} > \eta$, and $\rho_{m,N_\eta^m} \leq \eta$, $\rho_{m,N_\eta^m-1} > \eta$. Thanks to \fref{conveig}, we have $N_\eta^m \to N_\eta$ as $m \to \infty$ for each $\eta$, and we choose $m_0(\eta)$ sufficiently large that $|N_\eta-N_\eta^{m_0}| \leq 1$ for $m \geq m_0$. Then, for $m \geq m_0$ and proceeding as in the proof of Item 1, we find
\bee
\sum_{j \geq N_\eta^m} |\beta(\varrho_{m,j})| &\leq& E(\varrho_m)^\gamma \left( \sum_{j\geq N_\eta-1} \lambda_j^{-\frac{\gamma}{1-\gamma}}\right)^{1-\gamma}.
\eee
Since $N_\eta \to \infty$ as $\eta \to 0$, and $\sum_{j\in \Nm} \lambda_j^{-\frac{\gamma}{1-\gamma}}$ is finite as seen in the proof of Item 1, it follows that
$$
\lim_{\eta \to 0} \sup_{m \in \Nm} |\Tr\left(\beta^{(1)}_{\eta}(\varrho_m)\right|=0.
  $$
  Since the same estimate applies to $\beta^{(1)}_{\eta}(\varrho)$, we obtain the desired result by combining with \fref{combet2}. This ends the proof of Lemma \ref{lem:EntropyProp}.

  \subsection{Verification of Assumption \ref{A4}}

The proof is divided into various steps. We construct first the minimizer of the free energy $F_T(\varrho)$ under the global density constraint $\Tr(\varrho)=\bar n$. Then, we prove the strict monotonicity of the energy and the entropy as well as the upper and lower limits for the energy. 

 \paragraph{Preliminaries.}   Let
  $$
  \beta_-=\lim_{x \to 0} \beta'(x), \qquad \beta_+=\lim_{x \to \bar{n}} \beta'(x),
  $$
so that $\beta': (0, \bar{n}) \mapsto (\beta_-,\beta_+)$. We need to be careful with the range of $\beta'$ since the minimizer takes slightly different forms whether $\beta_-$ is finite or not. We have $\beta_-<\beta_+$ since $\beta'$ is strictly increasing, and possibly $\beta_-=-\infty$ and $\beta_+=+\infty$. We introduce \begin{equation*}
\xi:\; t\in(-\beta_+,-\beta_-) \to (\beta')^{-1}(-t)\in (0,\bar{n}),
\end{equation*}
with $\xi(-\beta_+):=\lim_{x \to -\beta_+} \xi(x)=\bar{n}$ and $\xi(-\beta_-):=\lim_{x \to -\beta_-} \xi(x)=0$ since $\xi$ is strictly decreasing.  For $T>0$, let $\xi_T(x)=\xi(x/T)$ and $\mu_0 \equiv \mu_0(T)=-T \beta_+-\lambda_0$ where $\lambda_0$ is the smallest eigenvalue of $H$ (and $\mu_0=-\infty$ by definition when $\beta_+=+\infty$). We will treat separately the cases $\beta_-$ finite and $\beta_-$ infinite since the proofs may differ, and the cases $\beta_+$ finite and $\beta_+$ infinite simultaneously since the proofs are identical. For uniformity of notation when $\mu_0=-\infty$ and $\beta_+=+\infty$, we will write $f(\mu_0)$ or $f(-\beta_+)$ for an arbitrary function $f$ as a shorthand for $\lim_{x \to -\infty} f(x)$.

We then introduce the following operator,  defined for $\mu \in [\mu_0,\infty)$, which will serve as candidate minimizer of $F_T$,
  $$
  \varrho_{T,\mu}:=\xi_T(H+\mu) \un_{\{H+\mu \leq -T \beta_-\}} = \sum_{j\in\mathbb{N}} \xi_T(\lambda_j+\mu) \un_{\{\lambda_j+\mu \leq -T \beta_-\}} \ket{\phi_j} \bra{\phi_j},
  $$
  where $\{\lambda_j\}_{j \in \Nm}$ and $\{\phi_j\}_{j \in \Nm}$ are the (nondecreasing sequence of) eigenvalues and eigenfunctions of $H$.
\begin{remark} Note that the choice $\mu \in [\mu_0,\infty)$ ensures that $H+\mu \geq -T \beta_+$. Furthermore, the cutoff $\un_{\{H+\mu \leq -T \beta_-\}}$ is necessary when $\beta_-$ is finite in order to make sure that the eigenvalues of $H+\mu$ remain in the range of $\beta'$. The cutoff becomes the identity when $\beta_-=-\infty$.
\end{remark} 
We introduce
\begin{equation}\label{eq:defN_T}
N_T(\mu) : = \max\left\{ j\in\mathbb{N}:\quad \lambda_{j}+\mu \leq -T \beta_-\right\},
\end{equation}
for  $\mu\in [\mu_0,\mu_M]$ with $\mu_M \equiv\mu_M(T):= -T \beta_- - \lambda_0 $. We extend its definition by setting $N_T(\mu) := -1$ when $\mu>\mu_M$ and, in the case where $\beta_-=-\infty$,  $N_T(\mu):=+\infty$ (and $\mu_M:=+\infty$). In particular, this yields, for $\mu\in [\mu_0,+\infty)$,
\begin{equation*}
\varrho_{T,\mu} = \sum_{j = 0}^{N_T(\mu)}  \xi_T(\lambda_j+\mu)\ket{\phi_j} \bra{\phi_j},
\end{equation*}
with the convention that a sum on an empty set is zero.
 \paragraph{Step 0: the operator $\varrho_{T,\mu}$ belongs to $\calE^+$.} It is clear that $\varrho_T$ is nonnegative. 
 
 \noindent\textit{The case $\beta_-$ infinite:} we then need to show that
  $$
  \Tr(\sqrt{H} \varrho_{T,\mu} \sqrt{H})+\Tr(\varrho_{T,\mu})=\sum_{j\in \Nm} (1+\lambda_j)\xi_T(\lambda_j+\mu)=:I<\infty.
  $$
  %where $\{\lambda_j\}_{j \in \Nm}$ is the nondecreasing sequence of eigenvalues of $H$.
  For this, we have $|\beta'(x)| \leq C_{\bar x,\gamma} |x|^{\gamma-1}$ for $x \in (0, \bar x]$ and $\gamma \in (\frac{d}{d+2},1)$ according to hypothesis \fref{polcont} in Theorem \ref{mainth2}. This implies that
  \be \label{contxi}
  |\xi(x)| \leq C_{\bar x,\gamma} |x|^{-\frac{1}{1-\gamma}}, \qquad \forall x \in [-\beta'(\bar x), -\beta_-).
  \ee
  For $\mu$ given in $[\mu_0,\infty)$, let $N_0$ such that $\lambda_j + \mu \geq -\beta'(\bar x)$ for all $j \geq N_0$, and let $N_1$ such that $\lambda_j \geq \mu/2$ for $j \geq N_1$. Define $N_2=\max(N_0,N_1)$. Since $\xi_T(\lambda_j+\mu)$ is a decreasing function of $\lambda_j+\mu$, we have, using \fref{contxi},
  \bea \nonumber
  I&\leq& \sum_{j=0}^{N_2-1} (1+\lambda_j)\xi_T(\lambda_0+\mu_0)
  +C_{\bar x,\gamma} T^{\frac{1}{1-\gamma}} \sum_{j \geq N_2} (1+\lambda_j) |\lambda_j+\mu|^{-\frac{1}{1-\gamma}}\\
  &\leq&\bar{n} \sum_{j=0}^{N_2-1} (1+\lambda_j)+ C'_{\bar x,\gamma}(2T)^{\frac{1}{1-\gamma}} \sum_{j \geq N_2} |\lambda_j|^{-\frac{\gamma}{1-\gamma}}. \label{contI}
    \eea
    Above, we used that $\xi_T(\lambda_0+\mu_0)=\bar n$. We have already seen in \fref{sumeig} that the last term on the right above is finite, and as a consequence $\varrho_{T,\mu} \in \calE^+$ for all $T>0$ and $\mu \in [\mu_0, \infty)$.

\noindent\textit{The case $\beta_-$ finite:} in that situation, it is clear that $\varrho_{T,\mu} \in \calE^+$ since $\varrho_{T,\mu}$ is of finite rank $N_T(\mu)$, and we obtain the same estimate for $I$ as above with  $C_{\bar x,\gamma}'=0$ and $N_2-1=N_T(\mu)$.
\paragraph{Step 1: Candidate for the minimizer of $F_T$.} In this paragraph, we construct a chemical potential $\mu_T$ such that $\Tr(\varrho_{T,\mu_T})=\bar n$, and as a consequence $\varrho_{T}:=\varrho_{T,\mu_T}$ belongs to the feasible set $\calS(\bar n)$ since we already know that $\varrho_{T} \in \calE^+$. Using the convexity of the entropy, we then show that $\varrho_T$ is the unique minimizer of $F_T$ in $\calS(\bar n)$.

Introduce the partition function
    \begin{equation*}
Z_T: \mu \in[\mu_0,+\infty) \to \Tr(\varrho_{T,\mu}) =  \sum_{j=0}^{N_T(\mu)}\xi_T(\lambda_j+\mu).
\end{equation*}
The function $Z_T(\mu)$ is strictly decreasing in $\mu$ since $\xi$ is strictly decreasing and $N_T(\mu)$ is nonincreasing, and we now prove the following lemma.
\begin{lemma}
For any $T>0$, the function $\mu\in[\mu_0,+\infty)\to Z_T(\mu)$ is continuous.
 \end{lemma}
 \begin{proof}
\noindent\textit{The case $\beta_-$ infinite:} this is a direct consequence of dominated convergence for series since $\mu \mapsto \xi_T(\lambda_j+\mu)$ is continuous and we have the uniform bound, for $\mu \in [\mu_-,\mu_+]\subset[\mu_0,+\infty)$,
\be \label{unixi}\xi_T(\lambda_j+\mu)\leq \xi_T(\lambda_j+\mu_-).\ee
\noindent\textit{The case $\beta_-$ finite:} we need to track how $N_T(\mu)$ changes when $\mu$ varies. Suppose first that $\mu$ is such that $\lambda_{N_T(\mu)}+\mu < -T \beta_-$. We can then take a sequence $\{\mu^k\}_{k \in \Nm}$ in a sufficiently small neighborhood of $\mu$ such that $N_T(\mu^k)=N_T(\mu)$ for all $k$, and continuity is direct. When  $\mu$ is such that $\lambda_{N_T(\mu)}+\mu = -T \beta_-$, consider the sequence $\mu^k=\mu+ \delta^k$ with $\lim_{k \to \infty} \delta^k=0$ and $\delta^k$ sufficiently small that, depending on the sign of $\delta^k$, either $N_T(\mu^k)=N_T(\mu)$ or $N_T(\mu^k)=N_T(\mu)-1$. When $\delta^k>0$, we have 
$$Z_T(\mu)-Z_T(\mu^k)=\sum_{j=0}^{N_T(\mu)-1}\left(\xi_T(\lambda_j+\mu)-\xi_T(\lambda_j+\mu^k)\right),$$ 
since the additional term in $Z_T(\mu)$ for $j=N_T(\mu)$ vanishes as $\xi_T(\lambda_{N_T(\mu)}+\mu)=\xi(-\beta_-)=0$. When $\delta^k<0$, we have the same expression as above with $N_T(\mu)-1$ replaced by $N_T(\mu)$. Since $\xi$ is continuous, this implies the continuity of $Z_T$ when $\beta_-$ is finite.
\end{proof}

We have all needeed now to construct a $\mu_T$ such that $Z_T(\mu_T)=\bar{n}$ for each $T\geq 0$.

\begin{lemma}
There exists a unique continuous function $T\in (0,+\infty)\to \mu_T$ such that, for any $T>0$,
\begin{equation*}
Z_T(\mu_T)=\bar{n}.
\end{equation*}
\end{lemma}
\begin{proof}
\noindent\textit{The case $\beta_-$ infinite:} on the one hand, we have, since $Z_T$ is strictly decreasing,
$$
Z_T(\mu) \leq Z_T(\mu_0)=\xi(-\beta_+)+\sum_{j=1}^{\infty}\xi_T(\lambda_j+\mu_0)=\bar{n}+C_T, \qquad C_T>0,
$$
with $C_T=+\infty$ (resp. finite) when $\beta_+=+\infty$ (resp. finite) by monotone convergence. Above, we used that $\xi(-\beta_+)=\bar n$ and that $\xi(t)>0$ when $t \in (-\beta_+,-\beta_-)$.
On the other hand, it follows from dominated convergence and \fref{unixi} that $\lim_{\mu \to \infty} Z_T(\mu)=0$ since $\lim_{\mu \to \infty} \xi_T(\lambda_j+\mu)=0$ for all $j \in \Nm$ and $T>0$. Hence, since $Z_T$ is continuously strictly decreasing on $[\mu_0,\infty)$ with values in $(0,\bar{n}+C_T]$, there exists a unique $\mu_T$ such that $Z_T(\mu_T)=\bar{n}$. Note that the version of the implicit function theorem for monotone functions shows that $T\mapsto \mu_T$ is continuous.

\noindent\textit{The case $\beta_-$ finite:} we carefully check the range of $Z_T(\mu)$ for $\mu \in [\mu_0, \mu_M]$, with now $\mu_M=-T \beta_- -\lambda_0$ finite. For this, define 
\be \label{Tc}
T_c : =\frac{\lambda_1-\lambda_0}{\beta_+-\beta_-}.
\ee
Suppose that $\beta_+$ is finite. We have $T_c>0$ since the ground state of $H$ is nondegenerate.  When $T>T_c$, we have $N_T(\mu_0) \geq 1$ since $\lambda_1+\mu_0=\lambda_1-\lambda_0 - T \beta_+ < - T \beta_-$. It follows that $\xi_T(\lambda_1+\mu_0)>0$ and we have $Z_T(\mu) \leq Z_T(\mu_0)=\bar{n}+C'_T$ for some $C_T'>0$. This provides us with an upper bound for $Z_T$. For the lower bound, we remark that $N_T(\mu_M) = -1$ when $\mu > \mu_M$, and that $\xi_T(\lambda_0+\mu_M)=\xi(-\beta_-)=0$, so that $Z_T(\mu_M)=0$. A similar argument based on monotonicity as in the case $N_T(\mu)=+\infty$ then goes through when $T>T_c$, and we obtain a unique $\mu_T$ continuous such that $Z_T(\mu_T)=\bar n$. When $T \leq T_c$, then only the $j=0$ mode contributes to $Z_T$. In that case, $Z_T(\mu_0)=\xi(-\beta_+)=\bar{n}$, and therefore, for all $T \leq T_c$, $\mu_T=\mu_0(T) = -T \beta_+-\lambda_0$. Now suppose that $\beta_+$ is infinite. Then $T_c=0$ and we proceed as in the case $T_c>0$.
\end{proof}

\bigskip

To summarize our results, we have constructed a unique continuous function $\mu_T$ such that $\varrho_T:=\varrho_{T,\mu_T} \in \calS(\bar n)$ for each $T>0$. Since $V$ is real, the eigenfunctions of $H+V$ can be chosen to be real-valued, and as a consequence $u[\varrho_T]=0$. We now prove that such $\varrho_T$ is the minimizer of $F_T$ in $\calS(\bar n)$.
\paragraph{Step 2: the operator $\varrho_T$ is the minimizer.} We show that
$$
F_T(\varrho) \geq F_T(\varrho_T), \qquad \forall \varrho \in \calS(\bar{n}).
$$
The proof is based on convexity and on the particular form of $\varrho_T$. We denote by $\{\rho_j\}_{j \in \Nm}$ the eigenvalues of $\varrho$ (with possibily $\rho_j=0$ for some $j \geq J$), and set 
\begin{equation*}
\nu_j: = \left\{\begin{array}{ll}
\xi_T(\lambda_j+\mu_T),\quad \textrm{for }0\leq j \leq N_T(\mu_T),
\\ 0,\quad \textrm{for }N_T(\mu_T)+1\leq j.
\end{array}\right.
\end{equation*} The convexity of $\beta$ yields
$$
\beta(\rho_j)\geq  \beta(\nu_j) + \beta'(\nu_j) (\rho_j-\nu_j), \qquad \forall j\in\Nm,
$$
which, together with estimate \fref{contsum}, the facts that $\beta'(0)=\beta_-$ and $T \beta'(\nu_j)=-\lambda_j-\mu_T$ for $j \leq N_T(\mu_T)$, gives
\begin{align*}
F_T(\varrho) &=  T \sum_{j \in \Nm} \beta(\rho_j)+ \sum_{j \in \Nm} \lambda_j \rho_j 
\\ &\geq  T\sum_{j \in \Nm} \left(\beta(\nu_j) + \beta'(\nu_j) (\rho_j-\nu_j) \right)+ \sum_{j \in \Nm} \lambda_j \rho_j\\
&\geq\sum_{j = 0}^{N_T(\mu_T)} (T\beta(\nu_j)+\lambda_j \nu_j)+\sum_{j = N_T(\mu_T)+1}^{+\infty} (T \beta_-+\lambda_j) \rho_j- \mu_T \sum_{j= 0}^{N_T(\mu_T)}(\rho_j-\nu_j)\\
&\geq F_T(\varrho_T)+\sum_{j = N_T(\mu_T)+1}^{+\infty} (T \beta_-+\lambda_j) \rho_j- \mu_T \sum_{j=0}^{N_T(\mu_T)}\rho_j + \mu_T \bar{n}.
\end{align*}
When $N_T(\mu_T)=+\infty$, the summation for $j \geq N_T(\mu_T)+1$ is conventionally equal to zero, and the last term becomes zero since both $\varrho$ and $\varrho_T$ are normalized with traces equal to $\bar{n}$. When $N_T(\mu_T)$ is finite, we have by construction that $\lambda_j+\mu_T \geq -T \beta_-$ for $j \geq N_T(\mu_T)+1$, and as a consequence
$$
\sum_{j = N_T(\mu_T)+1}^{+\infty} (T \beta_-+\lambda_j) \rho_j- \mu_T \sum_{j=0}^{N_T(\mu_T)}\rho_j \geq - \mu_T \sum_{j \in \Nm} \rho_j= - \mu_T \bar{n}.
$$
This results in $F_T(\varrho) \geq F_T(\varrho_T)$ for all $\varrho \in \calS(\bar{n})$, and since $F_T$ is strictly convex, it follows that $\varrho_T$ is the unique minimizer of $F_T$ in $\calS(\bar{n})$.

\bigskip

Now that we have obtained the minimizer $\varrho_T$ of $F_T$, our goal is to prove the strict monotonicity w.r.t the temperature of the energy and the entropy of $\varrho_T$, denoted by
$$
\mathrm{E}(T) = \Tr\left(\sqrt{H}\varrho_T\sqrt{H} \right), \qquad \mathrm{S}(T) = \Tr\left(\beta(\varrho_T)\right).
$$
For this, we will need to differentiate $\mathrm{E}(T)$ and $\mathrm{S}(T)$, which requires some regularity on $\xi$ and $\mu_T$. While we know that $\xi'$ exists and is continuous on $(-\beta_+,-\beta_-)$ since $\beta'' \in C((0,\bar n))$, we have no control at the end points and it is also unclear, without additional assumptions, how to make sure infinite sums involving $\xi'_T(\lambda_j+\mu_T)$ are finite. We therefore need to regularize $\varrho_T$ to justify the derivation. 

\paragraph{Step 3: Regularization.}

We treat the cases $\beta_-$ finite and infinite separately.

\noindent\textit{The case $\beta_-$ infinite:} for $m$ an integer and $\{\phi_j\}_{j \in \Nm}$ the eigenfunctions of $H$, consider the finite rank operator
$$
\varrho_m\equiv \varrho_{m,T,\mu}:=\sum_{j=0}^m \xi_T(\lambda_j+\mu) \ket{\phi_j} \bra{\phi_j}.
$$
Our first step is to construct a $\mu_{T,m}$ differentiable such that $\Tr(\varrho_{m,T,\mu_{T,m}})= \bar n$. Since the sum in $\varrho_m$ is finite, it is clear that $\varrho_m$ has a finite energy and entropy, and we need to make sure that its derivative w.r.t $T$ is well-defined, and in particular that it has finite energy. Since
\be \label{bxi}
\xi'(x)=-\frac{1}{\beta''(\xi(x))}, \qquad 0<\beta''\in C((0,\bar n)),
\ee
our goal is then simply to confine the range of $\xi$ to a compact set so that $\beta''(x) \geq C >0$ on this set. This is done by choosing the range of $\mu$ appropriately. For this, recall the notation $\mu_0=-T \beta_+-\lambda_0$ and that $\xi_T(\lambda_0+\mu_0)=\xi(-\beta_+)=\bar n$. Suppose first that $\beta_+$ is finite. For $T>0$, denote
$$
\delta n:=\sum_{j=1}^\infty \xi_T(\lambda_j+\mu_0)>0.
$$
Consider now $\delta_0>0$ such that $\delta_0 \leq \delta n/4$. Since $\xi$ is continuously strictly decreasing and so is $Z_T(\mu)$ w.r.t $\mu$, there exists $\eta>0$ such that 
$$
Z_T(\mu_0+\eta)=Z_T(\mu_0)-\delta_0=\bar{n}+\delta n-\delta_0.
$$
Set then $M_0$ such that 
$$
\sum_{j=0}^m \xi_T(\lambda_j+\mu_0+\eta) \geq \sum_{j=1}^\infty \xi_T(\lambda_j+\mu_0+\eta)-\frac{\delta n}{2} \qquad \textrm{for} \qquad m \geq M_0.
$$
% As a consequence, for $m \geq M_0$,
% $$
% \Tr(\varrho_{m,T, \mu_0})=\sum_{j=0}^m \xi_T(\lambda_j+\mu_0) \geq \bar{n}+\frac{\delta n}{2}.
% $$
% Consider now $\delta_0>0$ such that $\delta_0 \leq \delta n/4$. Since $\xi$ is continuously strictly decreasing, there exists $\eta>0$ such that 
% $$
% \xi_T(\lambda_0+\mu_0+\eta)=\xi(-\beta_++\eta /T)=\bar{n}-\delta_0.
% $$
Hence, for $\mu \leq \mu_0+\eta$ and $m \geq M_0$,
$$
\Tr(\varrho_{m,T,\mu})=\sum_{j=0}^m \xi_T(\lambda_j+\mu) \geq \sum_{j=0}^\infty \xi_T(\lambda_j+\mu_0+\eta)-\frac{\delta n}{2} \geq \bar{n}+\frac{\delta n}{4}.
$$
Next, it is clear that $\Tr(\varrho_{m,T,\mu}) \leq \Tr(\varrho_{T,\mu})$ since $\xi$ is nonnegative, and since both traces are decreasing functions of $\mu$, it is enough to look for $\mu_{T,m}$ in $[\mu_0+\eta, \mu_T]$ since such $\mu_{T,m}$ will necessarily be less than $\mu_T$ (we choose $\eta$ sufficiently small that $\mu_0+\eta<\mu_T$). Since
%$$
%\xi'(x)=-\frac{1}{\beta''(\xi(x))}, \qquad 0<\beta''\in C((0,\bar n)),
%$$
$\xi(x)$ lies in a compact set of $(0,\bar n)$ when $x \in [\mu_0+\eta, \mu_T]$, it follows from \fref{bxi} that $|\xi'|$ is bounded on $[\mu_0+\eta, \mu_T]$, and that
$$
Z_{T,m}(\mu):=\Tr\left(\varrho_{m,T,\mu}\right), \quad \mathrm{E}_m(T,\mu) := \Tr\big(\sqrt{H}\varrho_{m,T,\mu}\sqrt{H} \big), \quad \mathrm{S}_m(T,\mu) := \Tr\left(\beta(\varrho_{m,T,\mu})\right),
$$
are all continuously differentiable functions of $T$ and $\mu$ when $\mu \in [\mu_0+\eta, \mu_T]$.

Since $Z_{T,m}(\mu)$ is a continuously strictly decreasing function of $\mu$, and $Z_{T,m}(\mu_0+\eta)> \bar n$ as well as $Z_{T,m}(\mu_T)< \bar n$ by construction, it follows that there exists a unique $\mu_{T,m}$ such that $Z_{T,m}(\mu_{T,m})= \bar n$. A direct application of the implicit function theorem then shows that $T \mapsto \mu_{T,m}$ is continuously differentiable for $T>0$, and that
\begin{align} \nonumber
\partial_T \mu_{T,m} &= - \frac{\partial_T Z_{T,m}(\mu_{T,m})}{\partial_{\mu} Z_{T,m}(\mu_{T,m})} = \frac1T\frac{\Tr\left(\varrho_{T,m}'(H+\mu_{T,m})\right)}{\Tr(\varrho_{T,m}')} 
\\ &= \frac{\mu_{T,m}}T +  \frac1T\frac{\Tr\left(H\varrho_{T,m}'\right)}{\Tr(\varrho_{T,m}')}. \label{dmu}
\end{align}
Above, we used the notation $\varrho_{T,m}' = \xi_T'(H+\mu_{T,m})$ for $\xi_T'(x)=\xi'(x/T)$. Note that
$$
\Tr\left(H\varrho_{T,m}'\right)=\sum_{j=0}^m \lambda_j \xi'_T(\lambda_j+\mu_{T,m}), \qquad \Tr\left(\varrho_{T,m}'\right)=\sum_{j=0}^m \xi'_T(\lambda_j+\mu_{T,m}),
$$
are both well-defined since the sums are finite and $\xi'_T(\lambda_j+\mu_{T,m})$ is bounded for all $j=0, \ldots, m$. Writing $\mu_{T,m}=-\lambda_0  +T \gamma_{T,m}$, we find the equation
\begin{align} \label{eqgam}
\partial_T \gamma_{T,m} &= -\frac{\lambda_0}{T^2} +  \frac{1}{T^2}\frac{\Tr\left(H\varrho_{T,m}'\right)}{\Tr(\varrho_{T,m}')}.
\end{align}
Since $\lambda_0$ is the smallest eigenvalue of $H$, this shows that $\gamma_{m,T}$ is nondecreasing. The case $\beta_+=+\infty$ follows similarly, with the difference that $\mu \in [\mu_0+\eta,\mu_M]$ has to be replaced by $\mu \in [-\eta^{-1},\mu_M]$ for $\eta$ sufficiently small, we omit the details.

\noindent\textit{The case $\beta_-$ finite:} the minimizer is already finite rank, so we only need to verify that $\mu_T$ is not one of the  end points $\mu=\mu_0$ and $\mu=\mu_M$ since we have no control of $\xi'$ at these points. Suppose that $\beta_+$ is finite. When $T \leq T_c$, we have seen in Step 1 that only the $j=0$ mode contributes to the sum in $Z_{T}$. As a consequence, the minimizer is simply, for all $0<T\leq T_c$,
\begin{equation}\label{eq:rankonemin}
\varrho_{T}= \bar{n} \ket{\phi_0} \bra{\phi_0},
\end{equation}
and it is clear that the associated energy is constant for $0<T\leq T_c$. We therefore only consider the case $T > T_c$, with the case $T=T_c$ already solved. When $T>T_c$, we have $N_T(\mu_0) \geq 1$, and then
$$
Z_{T}(\mu_0)=\bar{n}+\sum_{j=1}^{N_T(\mu_0)} \xi_T(\lambda_j+\mu_0) > \bar{n}.
$$
Also, $Z_T(\mu_M)=0$, and since $Z_T(\mu)$ is strictly decreasing and continuous, it follows that for $T>T_c$, $\mu_T$ lies in a compact set $K$ of $(\mu_0,\mu_M)$. When $\beta_+$ is infinite, $T_c=0$ and $Z_T(\mu_0)=\infty$, and we proceed as in the case $T>T_c$. Since $\xi'$ is bounded on $K$, the function $Z_T(\mu)$ is then continuously differentiable for $T>T_c$ and $\mu \in K$, and a direct application of the implicit function theorem shows that $T \mapsto \mu_T$ is continuously differentiable. As a result, the quantities $Z_T(\mu_T)$, $\mathrm{E}(T)$, and $\mathrm{S}(T)$ are all continuously differentiable and the calculations of the next section are all justified.

%$ $\mu_M=-\lambda_0-\beta_- T$, and $\eta_1$ such that $\xi_T(\lambda_0+\mu_M-\eta_1)=\xi(-\beta_--\eta_1/T)=\bar{n}/2$.

\paragraph{Step 4: Monotonicity.}

We will prove the following Proposition:

\begin{proposition} \label{monot}
  For $T>0$, the energy $\mathrm{E}(T)$ (resp. the entropy $\mathrm{S}(T)$) is continuous and nondecreasing (resp. nonincreasing), and we have the relation for $T_2,T_1>0$,
  \be \label{eqF}
  F_{T_2}(\varrho_{T_2})-F_{T_1}(\varrho_{T_1})=\int_{T_1}^{T_2} S(\varrho_\tau) d\tau.
  \ee
\end{proposition}

\begin{proof} When $\beta_-$ is infinite, we work with the regularized minimizer $\varrho_{T,m}$ and then pass to the limit $m \to \infty$. The calculations are similar when $\beta_-$ is finite and are not detailed. Differentiating $\mathrm{E}_m(T)$ and using \fref{dmu}, we find, denoting $\varrho_{T,m}' = -|\varrho_{T,m}'|$ since $\xi'<0$,
   \begin{align*}
\partial_T \mathrm{E}_{m}(T) &= \sum_{j=0}^m\lambda_j\xi'_T(\lambda_j+\mu)\left( T\partial_T\mu_{T,m} - \lambda_j - \mu_{T,m}\right)/T^2
\\ &= \frac1{T^2 \Tr(\varrho_{T,m}')} \left(\left[\Tr\left( H\varrho_{T,m}'\right)\right]^2 - \Tr\left(H^2\varrho_{T,m}' \right)\Tr(\varrho_{T,m}')\right)
\\ &= \frac1{T^2 \Tr(|\varrho_{T,m}'|)} \left(\Tr\left(H^2|\varrho_{T,m}'| \right)\Tr(|\varrho_{T,m}'|) - \Tr\left(H |\varrho_{T,m}'|\right)^2 \right).
   \end{align*}
   Note that $\Tr\left(H^2|\varrho_{T,a_0,\varepsilon}'| \right)$ is well-defined since $|\varrho_{T,m}'|$ has finite rank. We now use the Cauchy-Schwarz inequality to deduce that
\begin{equation*}
\left[\Tr\left(H |\varrho_{T,m}'|\right)\right]^2 \leq \Tr\left(H^2|\varrho_{T,m}'| \right)\Tr(|\varrho_{T,m}'|),
\end{equation*}
which leads to the inequality
\begin{equation*}
\partial_T \mathrm{E}_{m}(T)\geq 0.
\end{equation*}
It follows that, for any $T_1\geq T_2>0$,
\begin{equation} \label{monreg}
\mathrm{E}_{m}(T_1)\geq\mathrm{E}_{m}(T_2).
\end{equation}
We now turn to the entropy. We remark first that
$$
\partial_T Z_T(\mu_T)= \sum_{j=0}^m \xi'_T(\lambda_j+\mu_{T,m})\left( T\partial_T\mu_{T,m} - \lambda_j - \mu_{T,m}\right)/T^2=0
$$
since $Z_T(\mu_T)=\bar n$. Then,
\begin{align} \nonumber
\partial_T \mathrm{S}_{m}(T) &= \sum_{j=0}^m \xi_T'(\lambda_j+\mu_{T,m})\left( T\partial_T\mu_{T,m} - \lambda_j - \mu_{T,m}\right) \beta'(\xi_T(\lambda_j+\mu))/T^2
  \\ &= - \sum_{j=0}^m \xi'_T(\lambda_j+\mu_{T,m})\left( T\partial_T\mu_{T,m} - \lambda_j - \mu_{T,m}\right)(\lambda_j+\mu_{T,m})/T^3\nonumber\\
  \nonumber&= - \sum_{j=0}^m \xi'_T(\lambda_j+\mu_{T,m})\left( T\partial_T\mu_{T,m} - \lambda_j-\mu_{T,m} \right)\lambda_j/T^3\\
  &= -\frac{1}{T}\partial_T \mathrm{E}_{m}(T). \label{monS}
\end{align}
This yields the monotonicity of the regularized entropy. With the above relations, it follows directly that
\be \label{derivF}
\partial_T F_T(\varrho_{T,m})= \mathrm{S}_{m}(T).
\ee
We now pass to the limit $m \to \infty$ for each $T$ fixed. We have seen in Step 3 that $\mu_{T,m} \in [f(\eta),\mu_T]$  for $f(\eta)=\mu_0+\eta$ when $\beta_+$ is finite, and $f(\eta)=-\eta^{-1}$ when $\beta_+$ is infinite. Hence, $\mu_{T,m}$ is uniformly  bounded in $m$. Denote by $\mu_\star$ the limit obtained by extraction of a subsequence. Since $\xi$ is continuously decreasing, we have
\be \label{contm}
\un_{j \leq m} \xi_T(\lambda_j+\mu_{T,m}) \leq \xi_T(\lambda_j+f(\eta)),
\ee
and using the latter for dominated convergence shows that
$$
\bar n=\lim_{k \to \infty} Z_{T,m_k}(\mu_{T,m_k})=Z_{T}(\mu_{\star}).
$$
Since for each $T>0$ there is a unique solution to the equation $Z_{T}(\mu)=\bar n$, it follows that $\mu_\star=\mu_T$, and also that the entire sequence $\mu_{T,m}$ converges to $\mu_T$. With \fref{contm}, we also obtain
\be \label{limE}
\lim_{m \to \infty} \mathrm{E}_{m}(T)=\mathrm{E}(T),
\ee
which, together with \fref{monreg}, yields the monotonicity of the energy.

For the continuity of $\mathrm{E}(T)$, we have seen in Step 1 that $T\mapsto \mu_T$ is continuous as an application of the implicit function theorem for strictly monotone functions. Note that we obtain as well from \fref{eqgam} that there exists a continuous nondecreasing function $\gamma_T$ such that
\begin{equation}\label{eq:defgamma_T}
\mu_T=-\lambda_0+T \gamma_T.
\end{equation}
Taking a sequence $\{T_n\}_{n \in \Nm}$ in $[T_-,T_+]$, we have the estimate, since $\xi_T(x)$ is an increasing function of $T$ for each $x$, 
\be \label{contm2}
\xi_{T_n}(\lambda_j+\mu_{T_n}) \leq \xi_{T_-}(\lambda_j+\bar \mu), \qquad \bar \mu=\max_{T \in [T_-,T_+]} \mu_T,
\ee
which allows us to pass to the limit in $\mathrm{E}(T_n)$ by dominated convergence. This yields the continuity of $\mathrm{E}(T)$.

We now treat the entropy, and we will use Item 2 of Lemma \ref{lem:EntropyProp} to pass to the limit in $\mathrm{S}_{m}(T)$ for $T$ fixed. We know from \fref{limE} that $\{\varrho_{T,m}\}_{m \in \Nm}$ is bounded in $\calE^+$, as well as $n[\varrho_{T,m}] V$ in $L^1(\Rm^d)$. It remains to show that $\varrho_{T,m}$ converges to $\varrho_T$ strongly in $\calJ_1$. For this, write
\bee
\varrho_T-\varrho_{T,m}&=&\sum_{j=m+1}^\infty \xi_T(\lambda_j+\mu_T) \ket{\phi_j} \bra{\phi_j}+\sum_{j=0}^m (\xi_T(\lambda_j+\mu_T)-\xi_T(\lambda_j+\mu_{T,m})) \ket{\phi_j} \bra{\phi_j}\\
&:=&\sigma_1+\sigma_2.
\eee
We have
$$
\|\sigma_1\|_{\calJ_1}=\sum_{j=m+1}^\infty \xi_T(\lambda_j+\mu_T),
$$
which goes to zero as $m \to \infty$ since $\varrho_T \in \calJ_1$. Besides,
$$
\|\sigma_2\|_{\calJ_1}=\sum_{j=0}^m|\xi_T(\lambda_j+\mu_T)-\xi_T(\lambda_j+\mu_{T,m})|,
$$
which also goes to zero thanks to \fref{contm} and dominated convergence. Item 2 of Lemma \ref{lem:EntropyProp} then gives
$$
\lim_{m \to \infty} \mathrm{S}_{m}(T)=\mathrm{S}(T),
$$
which, together with \fref{monS} after integration in $T$, yields the monotonicity of the entropy.

The continuity of $\mathrm{S}(T)$ follows from the same lines: we consider a sequence $\{\varrho_{T_n}\}_{n \in \Nm}$ for $T_n \in [T_-,T_+]$, which is bounded in $\calE^+$ with $n[\varrho_{T_n}] V$ bounded in $L^1(\Rm^d)$ as a consequence of \fref{contm2}. The fact that $\varrho_{T_n}$ converges to $\varrho_T$ in $\calJ_1$ is established in the same manner as above, with the additional ingredient consisting of the continuity of $T\mapsto \mu_T$. Item 2 of Lemma \ref{lem:EntropyProp} then yields the continuity of $\mathrm{S}(T)$ for $T>0$.

It remains to derive \fref{eqF} to conclude the proof. Integrating \fref{derivF}, we can pass to the limit in the $F_{T_i}(\varrho_{T_i,m})$ terms, $i=1,2$, using what was just done for $\mathrm{E}_m(T)$ and $\mathrm{S}_m(T)$, and we just need to treat the integral term. For this, we remark first that Item 1 of Lemma \ref{lem:EntropyProp} yields
$$
|\mathrm{S}_{m}(T)| \leq C+C E(\varrho_{T,m})^\gamma,
$$
and proceeding as in \fref{contI}, we find the estimate
 $$
  E(\varrho_{T,m})\leq \bar{n} \sum_{j=0}^{N_2(T)-1} (1+\lambda_j)+ C'_{\bar x,\gamma}T^{\frac{1}{1-\gamma}},
  $$
  where $N_2(T)$ is a finite integer. This allows us to use dominated convergence to pass to the limit in the integral term and obtain \fref{eqF}. This ends the proof.
\end{proof}

\bigskip

We now prove that the monotonicity of the previous Proposition is actually strict. We proceed by contradiction and suppose that $\mathrm{E}(T_1)=\mathrm{E}(T_2)$ for $T_1 \neq T_2$, with e.g $T_2>T_1>0$. Then, using \fref{eqF}, we have
\begin{equation*}
T_2\mathrm{S}(T_2) -T_1\mathrm{S}(T_1)= \int_{T_1}^{T_2}\mathrm{S}(\tau)d\tau.
\end{equation*}
Since $\mathrm{S}$ is nonincreasing, we have $\mathrm{S}(T_1) \geq \mathrm{S}(T_2)$, and the above equality gives
$$
T_2\mathrm{S}(T_2) -T_1\mathrm{S}(T_1) \geq (T_2-T_1)\mathrm{S}(T_2).
$$
This is equivalent to $\mathrm{S}(T_2)\geq \mathrm{S}(T_1)$, and as a consequence $\mathrm{S}(T_2)=\mathrm{S}(T_1)$. It follows that
\begin{equation*}
F_{T_1}(\varrho_{T_1}) = T_1 \mathrm{S}(T_1) + \mathrm{E}(T_1) = T_1 \mathrm{S}(T_2) + \mathrm{E}(T_2) = F_{T_1}(\varrho_{T_2}),
\end{equation*}
and, by the uniqueness of the minimizer, we deduce that $\varrho_{T_1} = \varrho_{T_2}$. This gives, for any $j\in\mathbb{N}$, since $\xi$ is one-to-one,
\begin{equation*}
\xi_{T_1}(\lambda_j+\mu_{T_1}) = \xi_{T_2}(\lambda_j+\mu_{T_2}) \Rightarrow \lambda_j\left(\frac{1}{T_1} - \frac{1}{T_2}\right) = \frac{\mu_{T_2}}{T_2} - \frac{\mu_{T_1}}{T_1}.
\end{equation*}
Since this equality must be true for any $j\in\mathbb{N}$, we deduce that $T_1 = T_2$ which contradicts our assumption. Therefore, $\mathrm{E}(T)$ is strictly increasing.

It remains to prove that $\mathrm{S}(T)$ is strictly decreasing. From \fref{eqF} and the fact that $\mathrm{S}(T)$ is nonincreasing, we have, for $T_2>T_1>0$,
\bee
0<\mathrm{E}(T_2) -\mathrm{E}(T_1)&=&T_1\mathrm{S}(T_1) -T_2\mathrm{S}(T_2) +\int_{T_1}^{T_2}\mathrm{S}(\tau)d\tau\\
& \leq& T_1\mathrm{S}(T_1)-T_2\mathrm{S}(T_2)+\mathrm{S}(T_1)(T_2-T_1)=T_2(\mathrm{S}(T_1)-\mathrm{S}(T_2)).
\eee
This finally gives $\mathrm{S}(T_1)>\mathrm{S}(T_2)$.

\paragraph{Step 5: The lower limit for the energy.} When $\beta_-=-\infty$, let $T_c=0$, and when $\beta_-$ is finite, define $T_c$ as in \fref{Tc}, with $T_c=0$ when $\beta_+=+\infty$. We prove that
$$
\lim_{T\to T_c} \mathrm{E}(T)=\lambda_0 \bar n.
$$
%The limit as $T \to T_c$ is direct to establish, while the limit as $T \to \infty$ is quite more involved. \medskip
We remark that, using \fref{ineqk} and $u[\varrho_T]=0$, for any $T>0$,
\begin{align}
 \mathrm{E}(T) &\geq \|\nabla \sqrt{n[\varrho_T]}\|^2_{L^2}+\|  n[\varrho_T] V\|_{L^1} = \langle H\sqrt{n[\varrho_T]}, \sqrt{n[\varrho_T]} \rangle\nonumber
\\ &\geq \lambda_0 \bar n \label{eq:ineqenergylmbd0}
\end{align}

\noindent\textit{The case $\beta_-$ finite and $\beta_+$ finite:} when $T\leq T_c$, we have already established that the minimizer is given by \eqref{eq:rankonemin} and, thus, we have directly $\mathrm{E}(T_c)=\lambda_0 \bar n$. 

\noindent\textit{The case $\beta_-$ infinite and $\beta_+$ finite or infinite:}  when $\beta_+$ is finite, since $\xi$ is decreasing,
$$
\mathrm{E}(T) \leq \sum_{j=0}^\infty \lambda_j \xi_T(\lambda_j+\mu_0)=\lambda_0 \bar n+\sum_{j=1}^\infty \lambda_j \xi((\lambda_j-\lambda_0)/T-\beta_+).
$$
Since $\xi(x) \to 0$ as $x \to + \infty$, it follows from dominated convergence (with dominating function e.g. $\lambda_j\xi(\lambda_j-\lambda_0-\beta_+)$) that the last term on the right converges to zero as $T \to 0$. The result then follows from \eqref{eq:ineqenergylmbd0}. Consider now the case where $\beta_+$ is infinite. Recall from the proof of Proposition \ref{monot} that there exists a continuous nondecreasing function $\gamma_T$ such that \eqref{eq:defgamma_T} holds, that is
$$\gamma_T = \frac{\mu_T+\lambda_0}T.$$ The function $\gamma_T$ therefore admits a limit as $T \to 0$. We will see that this limit has to be $-\infty$. Suppose the limit is finite and equal to $\gamma_\star$. Then, since $\xi$ is decreasing and $\gamma_T$ is nondecreasing,
\be \label{contn}
 \bar n=\sum_{j=0}^{\infty} \xi_T(\lambda_j+\mu_T) \leq \xi(\gamma_\star)+\sum_{j=1}^{\infty} \xi((\lambda_j-\lambda_0)/T+\gamma_\star).
\ee
The first term on the right is strictly less than $\bar n$ since $\xi(-\infty)=\bar n$ and $\gamma_\star$ is finite. The second term goes to 0 as $T \to 0$ by dominated convergence (with dominating function $\xi(\lambda_j-\lambda_0+\gamma_\star)$) since $\xi(x) \to 0$ as $x \to + \infty$. This yields a contradiction and therefore the limit of $\gamma_T$ as $T \to 0$ is $-\infty$. We now show that
\be \label{limlim} \lim_{T\to 0}
\frac{\lambda_1-\lambda_0}{T}+\gamma_T=+\infty.
\ee
The limit cannot by guessed directly since the first term goes to $+\infty$ and the second to $-\infty$. We proceed by contradiction and suppose that the limit is $L$, with $L<\infty$ (but possibly $L=-\infty$). We have, since $\xi \geq 0$,
$$
 \xi(\gamma_T)+\xi((\lambda_1-\lambda_0)/T+\gamma_T)\leq \bar n=\sum_{j=0}^{\infty} \xi_T(\lambda_j+\mu_T).
$$
Sending $T\to 0$, we find $\bar n+ \xi(L) \leq \bar n$. When $x \neq +\infty$, it follows that $\xi(x)>0$, and therefore that there is contradiction. Hence, \fref{limlim} holds, and there exists $T_1$ such that 
$$
\gamma_T \geq 1-\frac{\lambda_1-\lambda_0}{T}, \qquad \forall T \leq T_1.
$$ 
We then write, since $\xi$ is decreasing,
$$
\mathrm{E}(T) \leq \lambda_0 \xi(\gamma_T)+\lambda_1\xi((\lambda_1-\lambda_0)/T+\gamma_T)+\sum_{j=2}^\infty \lambda_j \xi((\lambda_j-\lambda_1)/T+1).
$$
The first term converges to $\lambda_0 \bar n$ since $\gamma_T \to -\infty$, the second and last one to zero because of \fref{limlim}, $\xi(+\infty)=0$, and dominated convergence with dominating function $\lambda_j \xi(\lambda_j-\lambda_1+1)$. This yields the expected limit thanks to \eqref{eq:ineqenergylmbd0}.

\noindent\textit{The case $\beta_-$ finite and $\beta_+$ infinite:} we have by construction $\gamma_T \leq -\beta_-$, with $\gamma_T$ nondecreasing. If $\gamma_T$ converges to a finite limit, then recalling \eqref{eq:defN_T}, we can see that $N_T(\mu_T) \to 0$ and we have
\begin{equation*}
 \bar n=\sum_{j=0}^{N_T(\mu_T)} \xi_T(\lambda_j+\mu_T) \leq \xi(\gamma_\star)+\sum_{j=1}^{N_T(\mu_T)} \xi((\lambda_j-\lambda_0)/T+\gamma_\star),
\end{equation*}
which leads to a contradiction by letting $T\to 0$ since $\xi(\gamma_{\star}) \leq \xi(-\infty) = \bar n$. Hence $\lim_{T \to 0} \gamma_T=-\infty$. We now show that 
\be \label{limlim2} \lim_{T\to 0}
\frac{\lambda_1-\lambda_0}{T}+\gamma_T=- \beta_-.
\ee
If this is not the case, we have for $T$ sufficiently small that $\lambda_1-\lambda_0+T \gamma_T < T\beta_-$ so that $N_T(\mu_T) \geq 1$. Following the same lines as \fref{limlim}, we find a contradiction using the fact that $\xi(x)>0$ if $x \neq \beta_-$. The limit \fref{limlim2} yields that $N_T(\mu_T) \to 1$. Then, since $\xi$ is nonincreasing
\bee
\mathrm{E}(T) &\leq& \lambda_0 \xi(\gamma_T)+\lambda_1\xi((\lambda_1-\lambda_0)/T+\gamma_T)+\sum_{j=2}^{N_T(\mu_T)} \lambda_j \xi((\lambda_j-\lambda_0)/T+\gamma_T)\\
&\leq &\lambda_0 \xi(\gamma_T)+\lambda_1\xi((\lambda_1-\lambda_0)/T+\gamma_T)+\sum_{j=2}^{N_T(\mu_T)}\xi(\gamma_T).
\eee
The first term converges to $\lambda_0 \bar n$ since $\gamma_T \to -\infty$, the second and last one to zero because of \fref{limlim2}, $\xi(\beta_-)=0$ and the fact that $N_T(\mu_T) \to 1$. This ends the proof thanks to \eqref{eq:ineqenergylmbd0}. 
%therefore $\varrho_T$ reduce to $\bar n \ket{\phi_0} \bra{\phi_0}$ in the limit $T \to 0$, giving the result. 

\paragraph{Step 6: The upper limit for the energy.} The upper limit as $T \to \infty$ is more subtle, and we use quite different arguments for the cases $\beta_-$ finite or infinite. We prove that 
\begin{equation}\label{eq:EnergyUpperLimit}
\lim_{T\to+\infty} E(T) = +\infty.
\end{equation}

\noindent\textit{The case $\beta_-$ infinite:} the proof combines relation \fref{eqF}, which allows us to quantify how the free energy changes with temperature, with the study of the behavior of $\mu_T/T$ as $T \to \infty$.  First of all, the convexity of $\beta$ yields $0=\beta(0) \geq \beta(x)- x \beta'(x)$ which, by using $\beta'(\xi_T(x))= -x/T$, gives for any $j\in\mathbb{N}$
$$
- (\lambda_j+\mu_T) \xi_T(\lambda_j + \mu_T)  \geq T \beta(\xi_T(\lambda_j+\mu_T)),
$$
and, by summation,
\begin{align*}
- \mu_T\bar n & = \sum_{j=0}^{+\infty}\xi_T(\lambda_j + \mu_T)\geq \sum_{j=0}^{+\infty} \lambda_j\xi_T(\lambda_j + \mu_T) +  T \beta(\xi_T(\lambda_j+\mu_T)) =  F_T(\varrho_T).
\end{align*}
As a consequence, we deduce the inequality
\be \label{blow1}
-  \bar n \mu_T  \geq F_T(\varrho_T)=E(\varrho_T) + TS(\varrho_T) \geq T \mathrm{S}(T).
\ee
Besides, for $T>T_1>0$, we deduce from \fref{eqF} and the decay of $\mathrm{S(T)}$ that
$$
\mathrm{E}(T) \geq F_{T_1}(\varrho_{T_1})-T\mathrm{S}(T)+\int_{T_1}^{T}\mathrm{S}(\tau) d\tau \geq F_{T_1}(\varrho_{T_1})-T_1 \mathrm{S}(T).
$$
Together with \fref{blow1}, this gives
\be \label{blow2}
\mathrm{E}(T) \geq F_{T_1}(\varrho_{T_1})+ T_1 \bar n \, \frac{\mu_T}{T}.
\ee
This is the key estimate. Indeed, we already know that $\mu_{T}/T=-\lambda_0/T+\gamma_T$ where $\gamma_T$ is nondecreasing, and as a consequence $ \mu_{T} / T$ is nondecreasing as well and therefore admits a limit as $T \to \infty$. Suppose first that this limit is finite and equal to $M$. Since $\xi$ is decreasing, we have
%\be \label{estxi}
$$
\xi_T(\lambda_j+\mu_T) \geq \xi(\lambda_j/T+M),
$$
%\ee
so that 
$$
Z_T(\mu_T)=\bar n \geq \sum_{j=0}^\infty  \xi(\lambda_j/T+M).
$$
Fatou's lemma then yields
\be \label{cont}
\bar n \geq \sum_{j=0}^\infty \liminf_{T\to \infty}\xi(\lambda_j/T+M) = \sum_{j=0}^\infty \xi(M) = +\infty,
\ee
giving a contradiction. Hence, $\mu_{T} / T \to \infty$ as $T \to \infty$. Then, fixing e.g. $T_1=1$ in \fref{blow2} shows that $\mathrm{E}(T) \to \infty$ as $T \to \infty$.

\noindent\textit{The case $\beta_-$ finite:} this case is actually more difficult than the previous one since we will see that $\mu_T/T$ converges to a finite number, and therefore \fref{blow2} is not enough to get the result. We will then have to resort to more technical tools and to a fine analysis of the high energy eigenvalues of $H$. We recall that  $N_T(\mu_T)$ is defined as the largest integer such that
\be \label{defal}
\lambda_{N_T(\mu_T)} \leq T (-\beta_--\mu_T /T) = : \alpha(T),
\ee
and that $\mu_T \leq \mu_M(T)= T \beta_--\lambda_0$. Our first task will be to show the following result.
\begin{lemma}
The following limits hold
\be \label{limalpha}
\lim_{T\to \infty} \alpha(T)=+\infty\quad\textrm{and}\quad \lim_{T \to \infty} \alpha(T)/T=0.
\ee
\end{lemma}
%For this, we show that 
%\begin{equation*}
%\mu_T/T\underset{T\to+\infty}{\sim}\mu_M(T)/T=-\beta_--\lambda_0/T.
%\end{equation*}
\begin{proof}
\noindent\textit{Step 1:} we begin with $\lim_{T \to \infty} \alpha(T)/T=0$. The limit of $\mu_T /T$ exists since it is nondecreasing as seen in the case $\beta_-$ infinite.  Suppose that $\mu_T /T$ converges to a constant $M<-\beta_-$, and therefore we have $\xi(M)>0$. Then $\alpha(T) \to +\infty$ according to its definition, and since $\xi$ is decreasing, 
$$
\bar n = \sum_{j=0}^{N_T(\mu_T)} \xi_T(\lambda_j+\mu_T) \geq  \sum_{j=0}^{N_T(\mu_T)} \xi(\lambda_j/T+M).
$$
Since $\alpha(T) \to +\infty$, we have $N_T(\mu_T) \to +\infty$ when $T \to +\infty$, and this leads to a contradiction according to Fatou's lemma:
$$
\bar n \geq \sum_{j=0}^\infty \liminf_{T\to \infty} \un_{j \leq N_T(\mu_T)} \xi(\lambda_j/T+M) = \sum_{j=0}^\infty \xi(M) =  +\infty.
$$
Hence, $\mu_T/T$ converges to $-\beta_-$ and $\lim_{T\to+\infty}\alpha(T)/T = 0$ as announced. 

\noindent\textit{Step 2:} we address now the limit $\lim_{T\to+\infty}\alpha : =\alpha_\infty$. It is undetermined since we have no information about the rate of convergence of $\mu_T /T$. Note that this limit exists as $\mu_T/T$ increases monotically to $-\beta_-$. There are then two possibilities: since $\alpha(T) \geq 0$, either $\alpha_\infty$ is infinite, or it is finite. Suppose the latter case holds. We then necessarily have $\alpha_\infty \geq \lambda_0$ from \fref{defal}, which shows that the limit of $N_T(\mu_T)$ is not empty and equal to some finite integer $N_\infty$ that depends on $\alpha_\infty$. Moreover, we have just seen that for $j$ fixed,
$$
\lim_{T\to \infty} \left(\frac{\lambda_j}{T}+\frac{\mu_T}{T}\right)=-\beta_-.
$$
Hence,
$$
\lim_{T\to \infty} \xi_T(\lambda_j+\mu_T) =\xi(-\beta_-)=0,
$$
which shows that
$$
\bar n = \lim_{T\to \infty} Z_T(\mu_T)=\sum_{j=0}^{N_\infty} \xi(-\beta_-)=0.
$$
There is a contradiction here, and as consequence $\alpha_\infty$ has to be infinite, and therefore $N_T(\mu_T)$ tends the infinity as $T\to \infty$.
\end{proof}

Owing to \fref{limalpha}, our goal is now to characterize the limit of $\mathrm{E}(T)$. This cannot be done by directly passing to the limit in the definition of the energy since $\xi_T(\lambda_j+\mu_T) \to \xi(-\beta_-)=0$ while the number of terms in the sum grows to infinity. We have the following result.

\begin{lemma}
There exists a constant $C>0$ such that
\begin{equation*}
\lim_{T\to+\infty} \frac{E(T)}{\alpha(T)} \geq C.
\end{equation*}
\end{lemma}
\begin{proof}
We remark first that for $y>0$ sufficiently small, we have from \fref{assbet},
\be \label{hold}
\xi_- y^{s} \leq \xi(-\beta_--y) \leq \xi_+ y^{s}, \qquad s=1/r>0,
\ee
with $\xi_\pm=c_\mp^{-s}$. By using \eqref{defal}, we rewrite $\xi_T(\lambda_j+\mu_T)$ as $\xi(-\beta_- - (\alpha(T)-\lambda_j)/T)$. Since we know that $\alpha(T)/T \to 0$ and that
$$
0 \leq (\alpha(T)-\lambda_j)/T \leq (\alpha(T)-\lambda_0)/T,
$$
we can choose $T$ sufficiently large so that $(\alpha(T)-\lambda_j)/T$ is arbitrary small for all $j \leq N_T(\mu_T)$. Then \fref{hold} yields
$$
\xi_- T^{-s} \sum_{j=0}^{N_T(\mu_T)} (\alpha(T)-\lambda_i)^s \leq \bar n \leq \xi_+ T^{-s} \sum_{j=0}^{N_T(\mu_T)} (\alpha(T)-\lambda_i)^s.
$$
The core to the proof is to estimate
\begin{equation*}
\sum_{j=0}^{N_T(\mu_T)} (\alpha(T)-\lambda_i)^s = \Tr\big( |\alpha(T)-H|^s_+ \big),
\end{equation*}
as $T \to \infty$. Above, $|x|_+$ denotes the positive part of $x$. Such quantities are called Riesz means in the literature (see e.g. \cite{HelfferRobert}). Their asymptotic behavior is well-known, and is obtained in terms of the quantity
$$
W_s(E):=\int_{\Rm^d} |E-V(x)|_+^{s+\frac{d}{2}} dx,
$$
leading to the following Weyl asymptotics, or limiting Lieb-Thirring equality,
\be \label{LT}
\lim_{E \to \infty} \frac{\Tr\big( |E-H|^s_+ \big)}{W_s(E)}= C_{s,d} : = \frac{\Gamma(s+1)}{(4 \pi)^{d/2}\Gamma(s+1+\frac{d}{2})},
\ee
where $\Gamma$ is the gamma function. The above relation is obtained by introducing the number of eigenvalues $\lambda_j[H]$ of $H=H_0+V$ less than $E$, defined by
$$
N(E,V)=\#\{ j \in \Nm: \; \lambda_j[H_0+V] \leq E\}.
$$
For the homogeneous potential $V_0(x)=|x|^\theta$, $\theta>0$, the following asymptotic formula holds for $N(E,V_0)$ (see \cite[Chapter 9 and 11]{RozBook})
$$ \ds
\lim_{E \to \infty }\frac{N(E,V_0)}{\ds \int_{\Rm^d} |E-V_0(x)|_+^{\frac{d}{2}} \, dx}= \frac{1}{(4 \pi)^{d/2}\Gamma(1+\frac{d}{2})}.
$$
 With our choice of potential $V(x)=1+V_0(x)$, a similar formula naturally holds by replacing $E$ by $E-1$. The result \fref{LT} then follows from the classical relation (see e.g. \cite{liebthirring})
$$
\Tr\big( |E-H|^s_+ \big)=s \int_0^{+\infty} y^{s-1} N(E-y,V_0) dy.
$$
We have now everything needed to conclude. 
By using \fref{LT}, we find that
$$
\lim_{T \to \infty }\frac{\bar n T^s }{\Tr\big( |\alpha(T)-H|^s_+ \big)}=\lim_{T \to \infty }\frac{\bar n T^s W_s(\alpha(T)) }{\Tr\big( |\alpha(T)-H|^s_+ \big) W_s(\alpha(T))}=\lim_{T \to \infty }\frac{\bar n T^s }{C_{s,d} W_s(\alpha(T))},
$$
and, with \fref{hold}, we find
$$
\xi_- \Tr\big( |\alpha(T)-H|^s_+ \big) \leq  \bar nT^s  \leq \xi_+ \Tr\big( |\alpha(T)-H|^s_+ \big),
$$
which leads to
$$
\frac{C_{s,d} \xi_-}{ \bar n} \leq \lim_{T \to \infty }T^s \left(W_s(\alpha(T))\right)^{-1} \leq \frac{C_{s,d} \xi_+}{ \bar n}
$$
This allows us to relate $T$ and $\alpha(T)$. Indeed, exploiting that $V_0(x)=|x|^\theta$, we have
$$
W_s(\alpha(T))=[\alpha(T)+1]^{s+(1+\frac{2}{\theta})\frac{d}{2}} \, W_s(1),
$$
which gives
\be\label{fi1}
\frac{ \bar n}{W_s(1) C_{s,d} \xi_+} \leq \lim_{T \to \infty }\frac{\alpha(T)^{s+(1+\frac{2}{\theta})\frac{d}{2}}}{T^s}\leq \frac{ \bar n}{W_s(1) C_{s,d} \xi_-}.
\ee
We now turn to the energy and use similar arguments. We have, for $T$ sufficiently large, 
\begin{equation}\label{eq:lowerTenergy}
G(T) : = \xi_- T^{-s} \sum_{j=0}^{N_T(\mu_T)} \lambda_j (\alpha(T)-\lambda_j)^s \leq \mathrm{E}(T),
\end{equation}
and rewrite the term on the left as
\begin{align*}
\sum_{j=0}^{N_T(\mu_T)} \lambda_j (\alpha(T)-\lambda_j)^s&=\alpha(T)\sum_{j=0}^{N_T(\mu_T)}(\alpha(T)-\lambda_j)^s-\sum_{j=0}^{N_T(\mu_T)} (\alpha(T)-\lambda_j)^{1+s}
\\ &= \alpha(T)  \Tr\big( |\alpha(T)-H|^s_+ \big) - \Tr\big( |\alpha(T)-H|^{s+1}_+ \big).
\end{align*}
Proceeding as above, we find
$$
\lim_{T\to+\infty}\frac{\alpha(T) \Tr\big( |\alpha(T)-H|^s_+ \big)}{[\alpha(T)]^{s+1+(1+\frac{2}{\theta})\frac{d}{2}}} =   C_{s,d}\, W_s(1),
$$
and
$$
\lim_{T\to+\infty}\frac{\Tr\big( |\alpha(T)-H|^{s+1}_+ \big)}{[\alpha(T)]^{s+1+(1+\frac{2}{\theta})\frac{d}{2}}} =   C_{s+1,d}\, W_{s+1}(1),
$$
In particular, this yields
\begin{equation*}
\lim_{T\to+\infty}\frac{\alpha(T) \Tr\big( |\alpha(T)-H|^s_+ \big)} {\Tr\big( |\alpha(T)-H|^{s+1}_+ \big)}=   \frac{C_{s,d}\, W_s(1)}{C_{s+1,d}\, W_{s+1}(1)} : = \kappa_{s,d},
\end{equation*}
where $\kappa_{s,d}>1$. Indeed, a simple calculation shows that
\begin{align*}
W_s(1) &= \int_{\mathbb{R}^d} |1- V_0(x)|_+^{s+d/2}dx = \omega_{d-1} \int_0^{+\infty} |1 - r^{\theta}|_+^{s+d/2} r^{d-1} dr
\\ &=\omega_{d-1} \int_0^1 (1-r^{\theta})^{s+d/2}r^{d-1}dr = \omega_{d-1} \frac{\Gamma\left(1+s+d/2\right)\Gamma\left(d/\theta\right)}{\theta\Gamma\left(1+s+(1+2/\theta)d/2\right)},
\end{align*}
where $\omega_{d-1}$ is the surface area of the $d-1$ sphere of radius $1$. It follows that
\begin{equation*}
\kappa_{s,d} = \frac{s+1+\frac{d}2}{s+1}\frac{1+s+\frac{d}{2}+\frac{d}{\theta}}{s+1+\frac{d}{2}}=\frac{1+s+\frac{d}{2}+\frac{d}{\theta}}{s+1}>1.
\end{equation*}
Finally, we find from \eqref{eq:lowerTenergy} and \eqref{fi1} that
\begin{align*}
\lim_{T\to+\infty} \frac{\mathrm{E}(T)}{\alpha(T)} &\geq \lim_{T\to+\infty} \frac{G(T)}{\alpha(T)} =\xi_-(\kappa_{s,d} -1) \lim_{T\to+\infty} \frac{\Tr\big( |\alpha(T)-H|^{s+1}_+ \big)}{T^s\alpha(T)}
\\ &\geq \xi_-(\kappa_{s,d} -1)C_{s+1,d}\, W_{s+1}(1)\lim_{T\to+\infty} \frac{\alpha(T)^{s+(1+\frac2{\theta})\frac{d}2}}{T^s}\geq \frac{\kappa_{s,d} -1}{\kappa_{s,d} }\frac{\xi_- \bar n}{\xi_+},
\end{align*}
which concludes the proof since $\kappa_{s,d}>1$.
\end{proof}

This ends the proof of \eqref{eq:EnergyUpperLimit} since $\alpha(T) \to \infty$ as $T\to \infty$. 
\begin{remark}
Note that we also have an upper bound which, thanks to  \fref{fi1}, shows that 
\begin{equation*}
\mathrm{E}(T)\underset{T\to+\infty}{\sim}\alpha(T) \underset{T\to+\infty}{\sim} T^{\frac{1}{1+(1+\frac{2}{\theta})\frac{d}{2s}}}.
\end{equation*}
\end{remark}

\bibliographystyle{plain}
\bibliography{bibliography.bib}

\end{document}